\documentclass[11pt,letterpaper]{article}

\usepackage[utf8]{inputenc}
\usepackage[T1]{fontenc}
\usepackage{times}
\usepackage[english]{babel}
\usepackage[letterpaper]{geometry}
\usepackage{setspace}
\usepackage{indentfirst}

\usepackage{amsmath, amsfonts, amssymb, amsthm, mathtools, actuarialsymbol}
\theoremstyle{definition}
\newtheorem{definition}{Definition}[section]

\newtheorem{remark}{Remark}[section]
\theoremstyle{plain}
\newtheorem{theorem}{Theorem}[section]

\newtheorem{corollary}[theorem]{Corollary}
\numberwithin{equation}{section}
\newtheorem{proposition}[theorem]{Proposition}

\usepackage{graphicx}
\usepackage{booktabs}
\usepackage{caption}
\usepackage{subcaption}
\usepackage[linesnumbered,ruled,vlined]{algorithm2e}
\usepackage{longtable}

\usepackage{authblk}
\usepackage{fancyhdr}
\usepackage[authoryear]{natbib}
\usepackage{hyperref}
\hypersetup{
    colorlinks=true,
    allcolors=blue,
    linktoc=all,
    pdfborder={0 0 0},
}
\usepackage[capitalise]{cleveref}

\title{Optimal annuitization with labor income under age-dependent force of mortality}

\author[1]{Criscent Birungi}
\author[1]{Cody Hyndman\thanks{Corresponding author. Emails: 
    \href{mailto:criscent.birungi@concordia.ca}{criscent.birungi@concordia.ca}, 
    \href{mailto:cody.hyndman@concordia.ca}{cody.hyndman@concordia.ca}}}

\affil[1]{Department of Mathematics and Statistics \\ Concordia University \\ Montr\'eal, QC, Canada}

\date{October 10, 2025}

\begin{document}

\maketitle

\begin{abstract}
\noindent We consider the problem of optimal annuitization with labour income, where an agent aims to maximize utility from consumption and labour income under age-dependent force of mortality. Using a dynamic programming approach, we derive closed-form solutions for the value function and the optimal consumption, portfolio, and labor supply strategies. Our results show that before retirement, investment behavior increases with wealth until a threshold set by labor supply. After retirement, agents tend to consume a larger portion of their wealth. Two main factors influence optimal annuitization decisions as people get older. First, the agent's perspective (demand side); the agent's personal discount rate rises with age, reducing their desire to annuitize. Second, the insurer's perspective (supply side); insurers offer higher payout rates (mortality credits). Our model demonstrates that beyond a certain age, sharply declining survival probabilities make annuitization substantially optimal, as the powerful incentive of mortality credits outweighs the agent's high personal discount rate. Finally, post-retirement labor income serves as a direct substitute for annuitization by providing an alternative stable income source. It enhances the financial security of retirees.

\vspace{1em} 

\noindent\textbf{Keywords:} Stochastic control; Optimal annuitization; Labor income; Optimal stopping; Force of mortality; Dynamic programming 
\end{abstract}

\raggedbottom

\section{Introduction}

\noindent An annuity is a contract that gives buyers a guaranteed, regular income. Annuitization converts a lump sum into regular income and is a critical financial decision, particularly for retirees. As \cite{buttarazzi2025optimal} notes, this is a retirement decision where individuals trade investment growth for long-term stability and guaranteed income for life. This decision is increasingly complex due to rising longevity risk, a demographic shift evidenced by growing labor force participation among older age groups \cite{gao2022optimal}. For instance, the Bureau of Labor Statistics projects that by 2033, the labor participation rate will exceed $10\%$ for those aged 75 and older and surpass $30\%$ for those aged 65. Many factors, such as labor income, contribute to post-retirement welfare, or bridge employment, which transforms conventional retirement planning \citep[see][]{klotz2021, lorenz2021, buttarazzi2025optimal}. The primary challenge is to determine how an agent can maximize lifetime utility from consumption and labor income, given age-dependent mortality risk, while managing wealth optimally.

This study addresses the agent's optimization problem within the framework of stochastic optimal control. The evolution of the agent's wealth is described by a stochastic differential equation (SDE) influenced by their strategic choices in investment ($\pi$), consumption ($c$), and labor ($b$). We employ the dynamic programming approach, specifically the Hamilton-Jacobi-Bellman (HJB) methodology, to solve this problem.
This approach, developed in the seminal works of \citet{merton1969lifetime, merton1975optimum} and \citet{karatzas1986explicit}, constitutes a central methodology for addressing such problems, complementing alternative frameworks such as the duality and martingale methods. Foundational contributions to the theory of duality were made by \citet{bismut1973conjugate}, while the comprehensive monograph by \citet{rockafellar1998variational} provides the authoritative modern treatment of optimization and duality theory. This work substantially extends the classical duality framework introduced in \citet{rockafellar1970convex}, generalizing it to encompass nonconvex and nonsmooth settings. Subsequent applications of these ideas include, among others, \citet{karatzas2000utility}, \citet{gerrard2012choosing}, and \citet{gao2022optimal}.

Our work builds on a rich literature exploring optimal consumption, leisure, and investment choices under various utility functions and constraints \citep[see][]{cvitanic1992convex, jin2006disutility, labbe2007convex, choi2008optimal, barucci2012optimal, koo2013optimal,  lee2015optimal, heunis2015quadratic, heunis2015utility, peng2023optimal, ferrari2023optimal}. A significant portion of the existing literature on optimal annuitization, including the works of \cite{lim2008optimal}, \cite{gerrard2012choosing}, and \cite{gao2022optimal}, simplifies the problem by assuming a constant mortality rate. This assumption implies that the agent's lifetime follows an exponential distribution with a constant mortality rate, which is far from realistic. By relaxing the constant mortality assumption, we extend the model of \cite{gao2022optimal}. Our objective is to incorporate a more realistic age-dependent force of mortality into the agent's utility maximization problem. Consequently, we proceed with a different objective function and approach to solve the problem. We adapt the dynamic programming framework from \cite{koo2013optimal}, providing rigorous proofs to derive closed-form solutions for the value function and the corresponding optimal consumption, labor, and portfolio allocation strategies. This approach avoids the guesswork often associated with developing auxiliary market models. Ultimately, we provide and implement results that have direct practical applications for retirement planning in an era of increasing longevity.

The rest of this paper is structured as follows. Section \ref{market_model} discusses the market model, including the economic background, mathematical formulations, and the stochastic optimal control problems. We propose methods for determining optimality and provide rigorous proofs to derive closed-form solutions for the value function associated with the optimal annuitization problem with labor income under an age-dependent force of mortality. Additionally, we present the main theoretical results for the optimal consumption, labor, and portfolio allocation strategies. Section \ref{Numerical_Implementation} presents numerical implementation results and key findings of our study. Finally, Section \ref{Conclusion_Recommendations} presents conclusions and recommendations for future research, and potential directions for further work.

\section{The Market Model}\label{market_model}
Consider a financial market consisting of two assets: a risk-free asset \( S_t^0 \), evolving as
\[
dS_t^0 = r S_t^0 \, dt,
\]
where \( r \) is the constant risk-free rate, and a risky asset \( S_t^1 \), following
\[
dS_t^1 = \mu S_t^1 \, dt + \sigma S_t^1 \, dW_t,
\]
where \( \mu \) is the expected return, \( \sigma \) is the volatility, and \( W_t \) is a standard Brownian motion. The agent’s wealth \( X_t \) is allocated between these two assets. Let
    \( \phi_t^0 \) be the units held in the risk-free asset,
 and    \( \phi_t^1 \) and units held of the risky asset.
Then, the portfolio value, or total wealth, is given by
\begin{equation}\label{total_wealth}
X_t = \phi_t^0 S_t^0 + \phi_t^1 S_t^1.
\end{equation}

We assume the portfolio is self-financing, that is
\begin{equation}\label{self_financing}
dX_t = \phi_t^0 dS_t^0 + \phi_t^1 dS_t^1.
\end{equation}
Substituting the dynamics of \( S_t^0 \) and \( S_t^1 \) into \eqref{self_financing}, we obtain
\begin{equation}\label{wealth_dynamics}
dX_t = r \phi_t^0 S_t^0 \, dt + \mu \phi_t^1 S_t^1 \, dt + \sigma \phi_t^1 S_t^1 \, dW_t.
\end{equation}

\noindent Let \( \pi_t = \phi_t^1 S_t^1 \) denote the amount invested in the risky asset. From \eqref{total_wealth}, we have \( \phi_t^0 S_t^0 = X_t - \pi_t \). Substituting these into \eqref{wealth_dynamics}, we obtain
\begin{align}\label{wealth_equation}
dX_t &= r (X_t - \pi_t) \, dt + \mu \pi_t \, dt + \sigma \pi_t \, dW_t \nonumber \\
     &= \left[ r X_t + \pi_t (\mu - r) \right] dt + \sigma \pi_t \, dW_t. 
\end{align}
The agent's strategic decisions are represented by the control processes for consumption, $c_t$, portfolio allocation, $\pi_t$, and labor, $b_t$. The agent is endowed with a total time $\overline{L}$, which is allocated between labor $b_t$ and leisure $l_t$, such that $b_t = \overline{L} - l_t$. We define $\tau$ as the voluntary retirement time, which is treated as an optimal stopping time. The agent's choice of leisure is therefore constrained as follows
\begin{equation}\label{eq:leisure_constraint}
l_{t} = 
\begin{cases} 
l_t \in [0, L] & \text{for } 0 \leq t < \tau \quad \text{(pre-retirement)} \\ 
\overline{L} & \text{for } t \geq \tau \quad \text{(post-retirement)}
\end{cases}
\end{equation}
where $0 \leq L < \overline{L}$ are constants, following the framework of \cite{choi2008optimal} and \cite{koo2013optimal}. This specification implies that labor $b_t$ can be positive only before retirement ($t < \tau$), and labor income ceases thereafter.

Incorporating these extensions into the wealth equation \eqref{wealth_equation}, then the agent's wealth, $X_t$, evolves according to the following stochastic differential equation (SDE)
\begin{equation}\label{wealth_process}
dX_t = \left[ r X_t + \pi_t (\mu - r) - c_t + w b_t \right] dt + \sigma \pi_t \, dW_t,
\end{equation}
subject to the standard non-negativity constraint on consumption, $c_t \geq 0$. The model parameters are defined as: $w$: the constant wage rate. The market price of risk is denoted by $\theta := \frac{\mu - r}{\sigma}$.\\

\noindent Given the probability space \( (\Omega, \mathcal{F}, P) \), as in \citet{lee2015optimal}, we assume the processes \( \pi_t \), \( c_t \), and \( b_t \) are all \( \mathcal{F}_t \)-progressively measurable and satisfy the following technical conditions
\begin{equation}\label{eqn:controls}
\int_0^t \pi_s^2 \, ds < \infty, \quad \int_0^t c_s \, ds < \infty, \quad \text{and} \quad \int_0^t b_s \, ds < \infty \quad \text{for all } t \geq 0 \text{ a.s.}
 \end{equation}

\noindent Accordingly, admissible control input \( \pi_t \) to the dynamical system \eqref{wealth_process} are always restricted to the real vector space
\begin{equation}\label{eqn:Pi}
\Pi := \left\{ \pi : [0, T] \times \Omega \to \mathbb{R}^N \,\middle|\, \pi \text{ is } \mathcal{F}_t\text{-prog. meas. and } \int_0^T \|\pi_t\|^2 \, dt < \infty, \text{ a.s.} \right\}.
\end{equation}

\noindent The wealth process \( X^{(\pi, c, b)} \) can be regarded as a stochastic dynamical system with the control input \( (\pi, c, b) \), where \( \pi \in \Pi \). To ensure that \( X^{(\pi, c, b)} \) is well-defined, the consumption process \( c \) must belong to the set
\begin{equation}\label{eqn: mathcal{C}}
\mathcal{C} := \left\{ c : [0, T] \times \Omega \to [0, \infty) \;\middle|\; c \text{ is } \mathcal{F}_t\text{-prog. meas. and } \int_0^T c_t \, dt < \infty \text{ a.s.} \right\}.
\end{equation}
Define $\mathcal{B}$, the set of admissible $labor$-processes defined by 
    \begin{equation}\label{eqn: mathcal{B}}
    \mathcal{B} := \left\{ b : [0, T]  \times \Omega \to [0, \infty) \;\middle|\; b \text{ is } \mathcal{F}_t\text{-prog. meas. and } \int_0^T b_t \, dt < \infty \text{ a.s.} \right\}.
    \end{equation}
Additionally, the wealth process \( X^{(\pi, c, b)} \) must not take negative values during the control interval \( [0, T] \). This ensures that the retiree cannot consume excessively, thereby avoiding negative wealth. A useful way to model these constraints is through a stipulated closed convex set \( K \subset \mathbb{R}^N \). %
The set of admissible control strategies is defined by the conditions detailed in the following remark.

\begin{remark}[Admissibility Conditions]\label{rem:admissibility}
For the optimization problem to be well-posed, we impose the following conditions on the control variables and the resulting wealth process.
\begin{enumerate}
    \item \textit{Control Set Constraints:}\label{item:control_constraints}
    The control variables $\pi(t, \omega)$ and $b(t, \omega)$ must belong to their respective admissible sets $K$ and $K_b$ for almost every $(t, \omega) \in [0,T] \times \Omega$. The set $K_b \subset \mathbb{R}^{N}$ is a closed convex set that constrains the labor choice $b_t$.

    \item \textit{Stopping Time:}\label{item:stopping_time}
    The retirement time $\tau$ must be an $\{\mathcal{F}_t\}$-stopping time. This ensures that the decision to retire at time $\tau$ depends only on information available up to that point, preventing any reliance on future knowledge.

    \item \textit{Finite Horizon:}\label{item:finite_horizon}
    The stopping time $\tau$ must occur on or before the terminal time $T$ almost surely.  That is, $\tau \leq T, \quad \mathbb{P}-\text{a.s.}$. This assumption imposes a finite time horizon on the control process.

    \item \textit{Wealth Admissibility:}\label{item:wealth_admissibility}
    A control tuple \( (\pi, c, b, \tau) \) is admissible if, for an initial wealth $X_0 > -\frac{w \overline{b}}{r}$, where \( \frac{w \overline{b}}{r} \) represents the present value of the agent's future labor income and the corresponding wealth process $X_t^{(\pi, c, b)}$ remains strictly greater than this lower bound for all $t \in [0, \tau]$. This is the minimum wealth or solvency threshold from which the agent's first behavior regime starts. It represents the amount that can be made up with labor income only and prevents the agent's wealth from falling below this threshold. The agent may consume and invest as long as initial wealth \( X_0 > -\frac{w \overline{b}}{r} \).
\end{enumerate}
\end{remark}

\noindent The set of admissible strategies, denoted by $\mathcal{A}$, consists of all control tuples $(\pi, c, b, \tau)$ that satisfy the conditions outlined in Remark \eqref{rem:admissibility}. We can formally define this set as:
\begin{equation}\label{admissible_control_inputs}
\mathcal{A} := \left\{ (\pi, c, b, \tau) \in \Pi \times \mathcal{C} \times \mathcal{B} \times \mathcal{T} \mid \text{Conditions \eqref{item:control_constraints}-\eqref{item:wealth_admissibility} are satisfied} \right\}
\end{equation}
Here, $\Pi$, $\mathcal{C}$, and $\mathcal{B}$ represent the sets of admissible investment, consumption, and labor processes, respectively. The set $\mathcal{T}$ contains all $\{\mathcal{F}_t\}$-stopping times $\tau$ such that $\tau \leq T$ $\mathbb{P}$-almost surely.

\begin{definition}\label{defn:admissible}
An investment-consumption-labor  strategy $(c, b, \pi, \tau)$ is said to be admissible if equation \eqref{eqn:controls} holds, almost surely.

\end{definition}
\noindent We now formally define the agent's utility maximization problem. The agent's preferences are governed by two concave utility functions: $u_1$, which captures the utility from consumption and labor before annuitization, and $u_2$, which represents the utility from wealth at the moment of annuitization. Following the framework of \cite{gerrard2012choosing}, the agent controls the consumption rate, the portfolio allocation, and the timing of the annuitization itself. A fund of size $x$ is converted into an annuity of $kx$, where the annuity factor $k>r$ and this decision is irreversible. If the fund is exhausted before this point, all economic activity, including further investment or withdrawal, ceases.

The agent derives utility \( u_1(c, b) \) from a payment of size \( c \) and  \( b \) before annuitization and \( u_2(kx_) \) from the annuity payment after annuitization. The agent seeks to maximize the sum of these expected discounted utilities over all \( (c, b, \pi, \tau) \in \mathcal{A} \). That is, we aim to determine \( (\pi^*, c^*, b^*, \tau^*) \in \mathcal{A} \) such that
\begin{equation}\label{eq:sup_value_01}
V(x) = \sup_{(c, b, \pi, \tau) \in \mathcal{A}} \mathbb{E} \left[\int_0^{T_D \wedge \tau} e^{-\beta t} u_1(c_t, b_t)) \, dt + \mathbf{1}_{ \{\tau < T_D\} } \int_{T_D \wedge \tau}^{T_D} e^{-\beta t} u_2(kX^{c, b, \pi, \tau}_{\tau}) \, dt \right],
\end{equation}
where,
    \( T_D \) is the retiree’s time of death, %
    \( \beta > 0 \) is a subjective discount rate that reflects the retiree’s time preference for money, \( kX^{c, b, \pi, \tau}_{\tau}\) is the annuity payment, and \(\mathcal{A}\) is the set of admissible control inputs.

Let \( n \) denote the retiree's current age, and let \( T_n \) represent the retiree's uncertain remaining lifetime. In \cite{gao2022optimal} and \cite{gerrard2012choosing}, it is assumed that \( T_n \) follows a distribution with a constant force of mortality \( \delta > 0 \). Under this assumption, the probability of surviving \( t \) years is given by
\begin{equation}\label{constant_mortality} 
P(T_n > t) = exp\left(-\delta t\right).
\end{equation}
However, a more realistic model would incorporate an age-dependent force of mortality, particularly one that increases with age. Following the argument in \cite{ashraf2023voluntary}, we introduce an age-dependent force of mortality and assume the Gompertz %
deterministic mortality law. That is,
\begin{equation}\label{eq:Gompertz}
\delta_t = \frac{1}{\lambda} exp\left(\frac{n_t - m}{\lambda}\right).
\end{equation}
where $n_t$ is the age at time $t$, $m$ is the modal value of life (in years), and $\lambda$ is the dispersion parameter (in years). In this case, the probability of surviving \( t \) years is expressed as
\begin{equation}\label{eq:age_dependent_mortality}
\px[t]{n}[(\delta)] =P(T_n > t)=exp\left(-\int_{0}^{t} \delta_y \, dy\right),
\end{equation}
where \( \delta_y \) represents the force of mortality at age \( y \). 
The deterministic age-dependent force of mortality is independent of the fund's level of evolution.  With this assumption, the objective function in \eqref{eq:sup_value_01} can now be expressed as
\begin{equation}\label{eq:sup_value_12.0}
V(x)= \sup_{(c, b, \pi, \tau) \in \mathcal{A}} \mathbb{E} \left[ \int_0^\tau e^{-\int_0^{s} (\beta + \delta_u) du} u_1(b_s, c_s) ds  + \frac{e^{-\beta \tau} \, }{\beta + \delta_\tau}  \px[\tau]{n}[(\delta)] u_2(kX_{\tau}) \right],
\end{equation} 
\begin{equation}\label{eq:sup_value_12}
= \sup_{(c, b, \pi, \tau) \in \mathcal{A}} \mathbb{E} \left[\int_0^\tau e^{-(\rho_s) s} u_1(b_s, c_s) ds  + \frac{e^{-(\rho_\tau) \tau}}{\rho_\tau} u_2(kX_{\tau}) \right].
\end{equation} 
where,   \( X_{\tau} = X_\tau^{x, c, b, \pi} \) is the wealth at the time of annuitization \( \tau \), $\rho_{t}$ is the effective discount rate given by
\begin{equation}\label{eq:age-dependent_force_of_mortality}
\rho_{t} = \int_{0}^{t} (\beta + \delta_y) \, dy.
\end{equation} 

\noindent In the remainder of the paper, we focus on equation \eqref{eq:sup_value_12}. 
Our analysis initially considers a constant effective discount rate 
$\rho$  in equation \eqref{eq:sup_value_12}. The more general scenario where $\rho_t$ varies with the mortality rate $\delta_t$, is considered when we derive the closed-form solutions for the value function in Theorem \ref{thm:value_function} and the optimal joint strategy $(\pi_t^*, c_t^*, b_t^*)$ in Theorem \ref{thm:optimal_policies}. Under the constant discount rate assumption, 
the objective function in \eqref{eq:sup_value_12} can be expressed as
\begin{equation}\label{eq:sup_value}
V(x ) = \sup_{(c, b, \pi, \tau) \in \mathcal{A}} \mathbb{E}\left[ \int_0^{\tau} e^{-\rho t} u_1(c_t, b_t) \, dt + \frac{e^{-\rho\tau}}{\rho} u_2(kX_{\tau}) \right].
\end{equation}

\noindent The operation of such a scheme may be subject to local regulation, as discussed in \cite{gerrard2012choosing} as 
    \( c_t \) and \( b_t \) may be restricted to lie in a given range \( (c_{\text{min}}, c_{\text{max}}) \)
    and \( (b_{\text{min}},  b_{\text{max}}) \)
    respectively, with both minimum and maximum values dependent on \( X_0 \).
    There may also be an upper limit on the annuitization time, for example, if retirees are required to purchase an annuity by a given age, the investment strategy \( \pi_t \) may be constrained to be non-negative and/or to be no greater than unity, depending on rules regarding the possibility of borrowing to fund additional equity purchases or regarding the short selling of risky assets.
However, in this paper, we treat only the situation of unconstrained controls \( c, b, \pi, and ~\tau  \).

The admissible control inputs in equation \eqref{admissible_control_inputs} are \( \mathcal{F}_t \)-progressively measurable processes 
such that the expectation in \eqref{eq:sup_value} is well-defined, i.e.,
\begin{equation} \label{eq:value_function_alt}
\mathbb{E} \left[ \int_0^\tau e^{-\rho t} u_1^-(c_t, b_t) \, dt + \frac{e^{-\rho \tau}}{\rho} u_2^-(kX_{\tau}) \right] < \infty,
\end{equation}
where $u^-$ denotes the negative part of $u$, that is, $u^- = \max(-u, 0)$. We say that the value $V(x)$ is \emph{attainable} if there exists a quadruple $(\pi^*, c^*, b^*, \tau^*) \in \mathcal{A}$ that achieves the supremum in \eqref{eq:sup_value}. If we let
\[
J(x; c, \pi, b, \tau) \triangleq \mathbb{E} \left[ \int_0^{\tau} e^{-\rho t} u_1(c_t, b_t) \, dt + \frac{e^{-\rho \tau}}{\rho} u_2(kX_{\tau}) \right],
\]
Then, the agent’s optimization problem is given by
\begin{equation}\label{eqn:V_x}
V(x) \triangleq \max_{(c, \pi, b, \tau) \in \mathcal{A}} J(x; c, \pi, b, \tau).
\end{equation}
We assume that \( V(x) < \infty \) for all \( x \in (0, \infty) \). Without loss of generality, the function \( V(\cdot) \) is increasing on \( (0, \infty) \). For the concavity properties of \( V(\cdot) \), see Section 8 of \cite{karatzas2000utility}. 
\begin{remark}
A sufficient condition for the assumption that \( V(x) < \infty \) in Equation \eqref{eqn:V_x} is that
\begin{equation}\label{eq:5.4}
\max\{u_1(x), u_2(x)\} \leq  \kappa_1 + \kappa_2 x^{\eta} \quad \forall x \in (0, \infty),\quad \quad \quad \quad \quad \quad \quad 
\end{equation}
for some \( \kappa_1 > 0,\quad\kappa_2 > 0,\quad\eta \in (0, 1) \), (see Remark 3.6.8 in \cite{karatzas1998methods}).
\end{remark}

\noindent Next, we specify the utility functions considered in equation \eqref{eq:sup_value_12}. 
The utility function with two elasticity parameters \( \alpha \) and \( \beta \) represents the preferences for consumption and labor, respectively. 
Following \cite{koo2013optimal}, the Cobb–Douglas utility function is defined as
\begin{equation}\label{3.2}
u(c, l) = u(c, b) = \frac{1}{\alpha} \cdot \frac{\left(c^\alpha b^{1-\alpha}\right)^{1-\gamma}}{1 - \gamma}, \quad 0 < \alpha < 1, \gamma > 0 \text{ and } \gamma \neq 1,    
\end{equation}
where \( \gamma \) is the retiree’s coefficient of relative risk aversion for two different goods \( c \) and \( b \), and \( \alpha \) is a constant parameter that measures the contribution of consumption to the agent’s utility when they aren't retired. If we define \( \gamma_1 := 1 - \alpha(1 - \gamma) \), the Cobb–Douglas utility function \eqref{3.2} can be rewritten as
\[
u(c, b) = \frac{c^{1-\gamma_1} b^{\gamma_1 - \gamma}}{1 - \gamma_1}.
\]
Then, the utility function \( u_1(c, b) \) in \eqref{eq:sup_value} becomes
\begin{equation}\label{3.2_main_appendix}
u_1(c, b) = \frac{c^{1-\gamma_1} b^{\gamma_1 - \gamma}}{1 - \gamma_1},
\end{equation}
where \( c \) is the consumption rate, \( b \) is the labor income rate, and \( \gamma_1 \) is the risk aversion coefficient for %
consumption. %
If the labor rate \( b = 1 \), the utility function \( u_1 \) becomes a power utility
\[
u_1(x) = \frac{x^{1 - \gamma_1}}{1 - \gamma_1}.
\]

\subsection{Dynamic Programming Principle}
\noindent We next present the Dynamic Programming Principle (DPP) for the value function. %
Let $\mathcal{T}_{t,T}$ denote the set of stopping times taking values in $[t, T]$. For a control process $\hat{\alpha} = (c, b, \pi, \tau)$ and a process $X_s$, define $X_s^{t,x}$ as the value of the state at time $s \geq t$ starting from $X_t = x$ under the control $\hat{\alpha}$ (so $X_t^{t,x} = x$). Let $\mathcal{A}(t,x)$ be the set of admissible control processes $\hat{\alpha} = (c, b, \pi, \tau)$ on the time interval $[t, T]$, then the value function $V(t,x)$ for the problem in \eqref{eq:sup_value} can now be expressed as
\begin{align}\label{eq:DPP_1}
V(t,x) = \sup_{\hat{\alpha} \in \mathcal{A}(t,x)} \mathbb{E} \left[ \int_t^\tau f(s, X_s^{t,x}, \hat{\alpha}_s)ds + g(\tau, X_\tau^{t,x}) \right],
\end{align}
where $f(s,X^{t,x}_{s}, \hat{\alpha}_{s}) = e^{-\rho s} u_1(c(s), b(s))$ is the running cost/utility rate, $g(\tau, X_\tau^{t,x})= \frac{e^{-\rho \tau}}{\rho} u_2(kX_\tau^{t, x}) $ is the terminal utility or cost at time $\tau$. 

\begin{theorem}[Dynamic Programming Principle]\label{thm:Dynamic_programming_principle}
Let $(t, x) \in [0, T] \times \mathbb{R}^N$, for any stopping time $\tau \in \mathcal{T}_{t,T}$, then we have 
\begin{align}\label{eq:DPP_2}
        V(t,x) &= \sup_{\hat{\alpha} \in \mathcal{A}(t,x)} \sup_{\tau \in \mathcal{T}_{t,T}} \mathbb{E} \left[ \int_t^\tau f(s, X_s^{t,x}, \hat{\alpha}_s)ds + g(\tau, X_\tau^{t,x}) \right] \\
        &= \sup_{\hat{\alpha} \in \mathcal{A}(t,x)} \inf_{\tau \in \mathcal{T}_{t,T}} \mathbb{E} \left[ \int_t^\tau f(s, X_s^{t,x}, \hat{\alpha}_s)ds + g(\tau, X_\tau^{t,x}) \right],
\end{align}
 with the convention that $e^{-\rho \tau(\omega)} = 0$ when $\tau(\omega) = \infty$.
\end{theorem}

\begin{remark}\label{remark:DPP}
In Theorem \ref{thm:Dynamic_programming_principle}, we shall often use the following equivalent formulation %
of the Dynamic Programming Principle
\begin{enumerate}
    \item[(i)] For all $\hat{\alpha} \in \mathcal{A}(t,x)$ and $\tau \in \mathcal{T}_{t,T}$:
    \begin{equation}
        V(t, x) \geq \mathbb{E} \left[ \int_t^\tau f(s, X_s^{t,x}, \hat{\alpha}_s)ds + g(\tau, X_\tau^{t,x}) \right]. 
    \end{equation}
    \item[(ii)] For all $\varepsilon > 0$, there exists $\hat{\alpha} \in \mathcal{A}(t,x)$ such that for all $\tau \in \mathcal{T}_{t,T}$
    \begin{equation}
        V(t, x) - \varepsilon \leq \mathbb{E} \left[ \int_t^\tau f(s, X_s^{t,x}, \hat{\alpha}_s)ds + g(\tau, X_\tau^{t,x}) \right]. 
    \end{equation}
\end{enumerate}
Then, a stronger version of the DPP for the finite horizon case in Remark \ref{remark:DPP} can be written as
\begin{equation}\label{eqn:DDP_strong}
    V(t, x) = \sup_{\hat{\alpha} \in \mathcal{A}(t,x)} \mathbb{E} \left[ \int_t^\tau f(s, X_s^{t,x}, \hat{\alpha}_s)ds + g(\tau, X_\tau^{t,x}) \right],
\end{equation}
for any stopping time $\tau \in \mathcal{T}_{t,T}$.
\end{remark}

\begin{proof}[Proof of Theorem \ref{thm:Dynamic_programming_principle}: Dynamic Programming Principle (DPP)]
We consider the finite horizon case. To prove the DPP in Theorem \ref{thm:Dynamic_programming_principle}, we proceed as follows.\\
\noindent\textbf{1.} Given an admissible control process $\hat{\alpha} \in \mathcal{A}_{t,T}$, 
for any stopping time $\tau$ valued in $[t, T]$, then by the law of iterated conditional expectation, we then get
\[ J(t, x, \hat{\alpha}) = \mathbb{E}\left[ \int_t^\tau f(s, X_s^{t,x}, \hat{\alpha}_s)ds + J(\tau, X_\tau^{t,x}, \hat{\alpha}) \right], \]
and since $J(\cdot, \cdot, \hat{\alpha}) \leq V$, and $\tau$ is arbitrary in $T_{t,T},$ then
\begin{align*}
J(t, x, \hat{\alpha}) &\leq \inf_{\tau \in \mathcal{T}_{t,T}} \mathbb{E} \left[ \int_t^\tau f(s, X_s^{t,x}, \hat{\alpha}_s)ds + g(\tau, X_\tau^{t,x}) \right] \\
&\leq \sup_{\hat{\alpha} \in \mathcal{A}(t,x)} \inf_{\tau \in \mathcal{T}_{t,T}} E \left[ \int_t^\tau f(s, X_s^{t,x}, \hat{\alpha}_s)ds + g(\tau, X_\tau^{t,x}) \right].
\end{align*}
By taking the supremum over $\hat{\alpha}$ in the left-hand side term, we obtain the inequality
\begin{equation}\label{ineq:step1_final_tau_latex}
    V(t,x) \leq \sup_{\hat{\alpha} \in \mathcal{A}(t,x)} \inf_{\tau \in \mathcal{T}_{t,T}} \mathbb{E} \left[ \int_t^\tau f(s,X^{t,x}_{s}, \hat{\alpha}_{s}) \, ds + g(\tau, X^{t,x}_{\tau}) \right].
\end{equation}

\noindent\textbf{2.} Fix some arbitrary control $\hat{\alpha} \in \mathcal{A}_{t,T}$ and $\tau \in \mathcal{T}_{t,T}$. By definition of the value functions in equation \eqref{eqn:DDP_strong}, for any $\varepsilon > 0$ and $\omega \in \Omega$, there exists $\hat{\alpha}^{\varepsilon,\omega} \in \mathcal{A}(\tau(\omega), X_{\tau(\omega)}^{t,x}(\omega))$, which is an $\varepsilon$-optimal control for $V(\tau(\omega), X_{\tau(\omega)}^{t,x}(\omega))$, i.e.
\begin{equation}\label{eqn:3.22}
V(\tau(\omega), X_{\tau(\omega)}^{t,x}(\omega)) - \varepsilon \leq J(\tau(\omega), X_{\tau(\omega)}^{t,x}(\omega), \hat{\alpha}^{\varepsilon,\omega}). 
\end{equation}
Define the process
\[ \alpha^{*}_s(\omega) =
\begin{cases}
\hat{\alpha}_s(\omega), & \text{if } s \in [0, \tau(\omega)] \\
\hat{\alpha}_s^{\varepsilon,\omega}(\omega), & \text{if } s \in (\tau(\omega), T].
\end{cases}
\]
The process $\alpha^{*}$ is progressively measurable and so lies in $\mathcal{A}(t,x)$ (see e.g., Chapter 7 in \cite{bertsekas1996stochastic}). Applying the law of iterated conditional expectation, and from equation \eqref{eqn:3.22}, we have
\begin{align*}
V(t, x) &\geq J(t, x, \alpha^{*}) = \mathbb{E} \left[ \int_t^\tau f(s, X_s^{t,x}, \hat{\alpha}_s)ds + J(\tau, X_\tau^{t,x}, \hat{\alpha}^{\varepsilon}) \right] \\
&\geq \mathbb{E} \left[ \int_t^\tau f(s, X_s^{t,x}, \hat{\alpha}_s)ds + g(\tau, X_\tau^{t,x}) \right] - \varepsilon.
\end{align*}
Since %
$\hat{\alpha} \in \mathcal{A}(t,x)$, $\tau \in \mathcal{T}_{t,T}$ and $\varepsilon > 0$, are arbitrary, we obtain the inequality
\begin{equation}\label{ineq:step2_final_tau_latex}
    V(t,x) \geq \sup_{\hat{\alpha} \in \mathcal{A}(t,x)} \sup_{\tau \in \mathcal{T}_{t,T}} \mathbb{E} \left[ \int_t^\tau f(s,X^{t,x}_{s}, \hat{\alpha}_{s}) \, ds + g(\tau, X^{t,x}_{\tau}) \right].
\end{equation}
By combining the two inequalities \eqref{ineq:step1_final_tau_latex} and \eqref{ineq:step2_final_tau_latex}, we obtain the required equality in equation \eqref{eqn:DDP_strong} for an arbitrary stopping time $\tau \in \mathcal{T}_{t,T}$. 
Thus, the proof of the Dynamic Programming Principle is completed. 
\end{proof}

\begin{remark}
For the optimization problem in \eqref{eq:sup_value} to be well-defined, several technical conditions must hold. First, we require positive annuity conversion factors, which emerge from the model parameters
\begin{equation}\label{eq:annuity_conversion_factor}
k := r + \frac{\rho - r}{\gamma} + \frac{\gamma - 1}{2 \gamma^2}  \theta^2 > 0 \quad \text{and} \quad k_1 := r + \frac{\rho - r}{\gamma_1} + \frac{\gamma_1 - 1}{2\gamma_1^2}  \theta^2 > 0.
\end{equation}
These conditions ensure the existence of optimal solutions and appear naturally when solving the Hamilton-Jacobi-Bellman equation in \eqref{eq:HJB_cont_region_rho_appendix}. This specific form reflects the trade-off between market parameters ($r$, $\theta$) and preference parameters ($\rho$, $\gamma$). 
\end{remark}

\begin{remark}
For the annuitization phase, we assume the value function $u_2(x)$ in \eqref{eq:sup_value_01}, which represents the agent's utility from an annuity purchased with wealth $x$, is analogous to the value function from the classical Merton problem \cite{koo2013optimal} and is obtained as follows

\begin{equation}
G(x) = \frac{e^{-\rho \tau}}{\rho} u_2(kX_{\tau}) =   e^{-\rho \tau} \frac{k^{1-\gamma_1}}{\rho(1-\gamma_1)} x^{1-\gamma_1}, \label{eq:post_retirement_value}
\end{equation}
which serves as our terminal condition. This formulation preserves the homothetic property of CRRA utility while incorporating the annuity conversion mechanism. %
\end{remark}

\begin{remark}
Define the quadratic function $f(m)$ for later consideration
\begin{equation}\label{eq:quadratic_main_appendix}
    f(m) := \frac{1}{2}\theta^2 m^2 + \left(\rho - r + \frac{1}{2}\theta^2\right)m - \rho = 0.
\end{equation}
whose distinct roots $m_+ > 0$ and $m_- < -1$ play crucial roles in determining the optimal strategies. The transformed quantities $p_1=m_1+1$ and $p_2=m_2+1$ are free boundary parameters associated with optimal retirement timing.
\end{remark}

\begin{definition}\label{defn:optimal_retirement_wealth_threshold}
The \textit{optimal retirement wealth threshold}, denoted by $x^*$, is the critical level of wealth at which an agent chooses to retire to maximize their lifetime utility.
\end{definition}
\begin{definition}\label{defn:subsistence_consumption_wealth_threshold}
The \textit{subsistence consumption wealth threshold}, denoted by $\tilde{x}$, is the level of wealth required to meet the minimum consumption boundary $\tilde{c}$ when the labor income component $b$ reaches its upper limit $\overline{b}$.
\end{definition}

\begin{theorem}[Value function]\label{thm:value_function}
Given that \( X_t = x \), assume the regularity conditions of Definition~\ref{defn:admissible} hold, then the value function \( V(x) \) is given by

\begin{equation}\label{eq:value_function}
V(x) =
\begin{cases}
V_{int}(x), & \text{if } 
 x < \tilde{x}, \\
V_{\bar{b}}(x), & \text{if } \tilde{x} \leq x < x^*, \\
G(x), & \text{if } x \geq x^*,
\end{cases}
\end{equation}
\noindent where $ G(x)$ is terminal condition defined in \eqref{eq:post_retirement_value}, with $\rho_{t}$, the effective discount rate given in equation~\eqref{eq:age-dependent_force_of_mortality}, which depends on the age-dependent force of mortality $\delta_{y}$ defined in equation \eqref{eq:Gompertz}. The functions $V_{\text{int}}(\cdot)$ and $V_{\overline{b}}(\cdot)$ are the interior value functions for $x < \tilde{x}$ and $\tilde{x} \leq x < x^*$, respectively.
\end{theorem}
\begin{proof}[Proof of Theorem \ref{thm:value_function}] See Appendix~\ref{Proof_of_Theorem_1:_Value_Function}. 
\end{proof}

\begin{proposition}[Optimal retirement threshold]\label{prop:optimal_threshold}
Let $V_{\bar{b}}(x)$ and $G(x)$ be the value functions for an agent in the final working phase and the retired state, respectively, as defined in Theorem \ref{thm:value_function}. The optimal retirement wealth threshold, $x^*$, is determined by the \textit{smooth-pasting condition} at the boundary, which ensures that the marginal value of wealth is continuous, that is $\left(
\frac{d V_{\bar{b}}}{d x}(x^*) = \frac{d G}{d x}(x^*)\right)$. Then, the threshold $x^*$ is the unique solution to the following non-linear algebraic equation
\begin{align}\label{eq:threshold_condition}
&\frac{\frac{1}{2} \theta^2 (m_+ - m_-)}{\beta + \delta_t} B_1 k^{-\gamma_1 m_+} \left( \frac{b}{\bar{b}} \right)^{m_+ (\gamma_1 - \gamma)} (x^*)^{-\gamma_1 m_+} \nonumber \\
&= \left[ \frac{r - \frac{1}{2} \theta^2 m_-}{\beta + \delta_t} - \frac{1}{1 - \gamma_1} \left( 1 - \left( \frac{b}{\bar{b}} \right)^{-\frac{\gamma_1 - \gamma}{\gamma_1}} \right) \right] x^* + \frac{r - \frac{1}{2} \theta^2 m_-}{\beta + \delta_t} \frac{w (\bar{b} - b)}{r}
\end{align}
\end{proposition}
\begin{proof}[Proof of Proposition \ref{prop:optimal_threshold}] See Appendix~\ref{app:proof_x_star}.
\end{proof}

\begin{definition}\label{defn:threshold_x_tilde}
Let $V_{\text{int}}(x)$ and $V_{\bar{b}}(x)$ be the value functions for an agent in the final working phase and the retired state, respectively, as defined in Theorem~\ref{thm:value_function}. The boundary between these states is a unique wealth threshold, $\tilde{x}$, determined by the \textit{value-matching condition}. This condition requires the two functions to be equal at the boundary
\begin{equation}
 V_{\text{int}}(\tilde{x}) = V_{\bar{b}}(\tilde{x}) 
\end{equation}
\end{definition}

\begin{proposition}[Wealth threshold for the labor constraint]\label{prop:threshold_x_tilde}
Applying the value-matching condition from Definition~\ref{defn:threshold_x_tilde} yields the following system of equations, which implicitly defines the wealth threshold $\tilde{x}$ and the corresponding consumption level $\tilde{c}$:
\begin{align}
  \tilde{x} &= A_2 \tilde{c}^{-\gamma m_-} + \frac{1}{\alpha k} \tilde{c} - \frac{w \bar{b}}{r} \label{eq:xtilde_region1} 
  = B_1 \tilde{c}^{-\gamma_1 m_+} + B_2 \tilde{c}^{-\gamma_1 m_-} + \frac{1}{k_1} \tilde{c} - \frac{w (\bar{b} - b)}{r} %
\end{align}
\end{proposition}
\begin{proof}[Proof of Proposition~\ref{prop:threshold_x_tilde}] See Appendix~\ref{Proof_of_Derivation_of_tilde{c}_}.
\end{proof}

\subsubsection{Hamilton-Jacobi-Bellman Variational Equality}
\noindent We now employ a dynamic programming approach within the Hamilton-Jacobi-Bellman (HJB) framework to determine an optimal $(\pi^*, c^*, b^*, \tau^*) \in \mathcal{A}$ that satisfies \eqref{eqn:V_x}. Consider the value function at time $t$ with wealth $x$, 

\begin{equation}
J(t, x) = \sup_{(c, b, \pi, \tau) \in \mathcal{A}} \mathbb{E} \left[ \int_t^{\tau} e^{-\rho t} u_1(c_t, b_t) \, dt + \frac{e^{-\rho \tau}}{\rho} u_2(kX_{\tau}) \right]. \label{eq:dynamics}
\end{equation}
such that,
\begin{align*}
    dX_t &= \left[rX_t + \pi_t(\mu - r) - c_t + wb_t\right]dt + \sigma \pi_t \, dW_t,
\end{align*}
and 
\begin{equation}\label{eq:labor_income_rate_2}
l_{t} = \begin{cases} 
0 \leq l_t \leq L & \text{if } 0\leq t < \tau \quad  \\ 
\overline{L} & \text{if } t \geq \tau \quad 
\end{cases}
\end{equation}
where
\[
J(\tau, x) = \frac{e^{-\rho \tau}}{\rho} u_2(kX_{\tau}),
\]

\noindent is the terminal condition at time \( \tau \). Let $V(x)$ be the value function, $V(x) = J(x)$. Let 
$G(x) = \frac{e^{-\rho \tau}}{\rho} u_2(kX_{\tau})$ %
represent the value obtained upon fully retiring at wealth $x$ that \( X_t = x \). Then the value function $V(x)$ satisfies the HJB variational equality %
\begin{equation} \label{eq:hjb_variational_rho_appendix}
\max \left\{ \sup_{c\ge0, 0 \leq b \leq \bar{b}, \pi} \left[ u_1(c, b) + \mathcal{L}^{c, b, \pi, \tau} V(x) \right] - \rho_{t}\, V(x), \quad G(x) - V(x) \right\} = 0,
\end{equation}
where the generator $\mathcal{L}^{c, b, \pi, \tau} V(x)$ is
\[
\mathcal{L}^{c, b, \pi, \tau} V(x) = \left( r x + \pi (\mu - r) - c + w b \right) V'(x) + \frac{1}{2} \sigma^2 \pi^2 V''(x).
\]
The state space is divided into a continuation (working) region ($x < {x^*}$), where $V(x) > G(x)$, and the HJB equation holds
    \begin{equation}\label{eq:HJB_cont_region_rho_appendix}
    \rho_{t} V(x) = \sup_{c\ge0, 0 \leq b \leq \bar{b}, \pi} \left[ u_1(c, b) + \left( r x + \pi (\mu - r) - c + w b \right) V'(x) + \frac{1}{2} \sigma^2 \pi^2 V''(x) \right],
    \end{equation}
and 
a stopping (retirement) region ($x \ge x^*$, where $V(x) = G(x)$. The value ${x^*}$ is the optimal retirement wealth threshold.

\begin{remark}[Structure of the optimal strategy]\label{rem:optimal_strategy_structure}
The HJB variational inequality in \eqref{eq:hjb_variational_rho_appendix} partitions the state space into two distinct regions, separated by the optimal retirement wealth threshold $x^*$.

\begin{enumerate}
    \item \textit{Continuation Region ($x < x^*$):} For wealth levels below the threshold, it is optimal for the agent to continue working. In this region, the value function $V(x)$ satisfies the HJB equation, as detailed in Proposition \ref{prop:optimal_threshold}:
    \[ 
    \rho_{t} V(x) = \sup_{c, b, \pi} \left[ u_1(c, b) + \mathcal{L}^{c, b, \pi} V(x) \right] 
    \]

    \item \textit{Stopping (Retirement) Region ($x \ge x^*$):} Once the agent's wealth reaches or exceeds the threshold, it is optimal to stop working and retire. The value function is then equal to the retirement value function, $G(x)$:
    \[ 
    V(x) = G(x) 
    \]
    where $G(x) = \frac{e^{-\rho \tau}}{\rho} u_2(kX_{\tau})$ represents the total utility derived from retiring with wealth $x$.
\end{enumerate}
\end{remark}

\begin{theorem}[Optimal policies in the continuation region]\label{thm:optimal_policies}
Let $V(x)$ be a twice-continuously differentiable solution to the Hamilton-Jacobi-Bellman (HJB) equation \eqref{eq:hjb_variational_rho_appendix} in the continuation region, where $x < x^*$. Then the admissible optimal policies for portfolio allocation, $\pi^*(x)$, consumption, $c^*(x)$, and labor, $b^*(x)$, are given by the following first-order conditions

\begin{enumerate}
    \item \textit{Optimal portfolio:} The optimal fraction of wealth invested in the risky asset follows the Merton rule
    \begin{equation}
        \pi^*(x) = -\frac{\theta}{\sigma} \frac{V'(x)}{V''(x)}. \label{eq:foc_pi_main_appendix}
    \end{equation}

    \item \textit{Optimal consumption:}\label{Optimal_Consumption} The optimal consumption rate $c^*$ is determined by the condition that the marginal utility of consumption equals the marginal value of wealth
    \begin{equation}\label{eq:optimal_consumption_policy_bmn}
        \frac{\partial u_1(c^*, b^*)}{\partial c} = V'(x).
    \end{equation}

    \item \textit{Optimal labor:}\label{Optimal_Labor} For an interior solution where $0 < b^* < \bar{b}$, the optimal labor supply $b^*$ satisfies the condition that the marginal disutility of labor equals its marginal contribution to the value of wealth
    \begin{equation}
        \frac{\partial u_1(c^*, b^*)}{\partial b} = -w V'(x). \label{eq:foc_b_main_appendix}
    \end{equation}
    \noindent If condition \eqref{Optimal_Labor} is not met for any $b \in (0, \bar{b})$, the optimum is a corner solution at the boundary, $b^* = \bar{b}$.

    \item \textit{Marginal rate of substitution:}\label{marginal_rate_of_substitution} Combining the conditions for Optimal consumption and labor in \eqref{Optimal_Consumption} and  \eqref{Optimal_Labor}, respectively, for an interior solution yields the classic static optimality condition, where the marginal rate of substitution between consumption and leisure equals the wage rate
\begin{equation}
    \frac{\partial u_1(c^*, b^*) / \partial b}{\partial u_1(c^*, b^*) / \partial c} = -w. \label{eq:MRS_condition_main_appendix}
\end{equation}
\end{enumerate}
\end{theorem}
\begin{proof}[Proof of Theorem \ref{thm:optimal_policies}] See Appendix~\ref{Proof_of_Theorem_3_Optimal_Policies_main_final}.
\end{proof}

\begin{corollary}[Actuarially fair annuity rate]\label{cor:actuarially_fair_annuity_rate}
The market-consistent annuity rate, denoted $k_{annuity}^*$ %
can be given in an actuarially fair framework by
\begin{equation}
k_{annuity}^* = \frac{y_{}}{\ddot{a}(\beta,\delta_{y})},
\end{equation}
where,
\begin{equation}
\ddot{a}(\beta,\delta_{y}) = \int_0^\infty e^{-\beta s} \, { \px[s]{y}[(\delta)] } \, ds,
\end{equation}
and $\beta > 0$ is a subjective discount rate that reflects the retiree’s time preference for money, $\px[t]{y}[(\delta)]$ is the survival probability defined in equation \eqref{eq:age_dependent_mortality} for an individual currently aged $y_{}$ to survive $s$ more years, and $\delta_{y}$ represents the age-dependent force of mortality for that individual, as defined in equation \eqref{eq:Gompertz}.
\end{corollary}

\begin{theorem}[Optimal policies]\label{thm:optimal_policies_results}
Assume the value function $V(x)$ from Theorem~\ref{thm:value_function} exists and satisfies the regularity conditions of Definition~\eqref{defn:admissible}, then the optimal strategies are given by
\begin{align}
    c_t^* &=
    \begin{cases}
        \left( \dfrac{V'_{\mathrm{int}}(x)}{\left(-\dfrac{1-\alpha}{w\alpha}\right)^{(1-\alpha)(1-\gamma)}} \right)^{-1/\gamma}, & \text{if } x < \tilde{x}, \\[1ex]
        \left( \dfrac{V'_{\bar{b}}(x)}{(\bar{b}^{1-\alpha})^{1-\gamma}} \right)^{-1/\gamma_1}, & \text{if } \tilde{x} \leq x < x^*, \\[1ex]
        \rho_{t}^{1/\gamma_1} k^{(\gamma_1 - 1)/\gamma_1} x, & \text{if } x \geq x^*,
    \end{cases}
    \label{eq:optimal_consumption_policy_aligned}
    \end{align}
\begin{align}
    b_t^* &=
    \begin{cases}
        \left( \dfrac{1 - \alpha}{w\alpha} \right) \left( \dfrac{k^{1 - \gamma_1}}{\beta + \delta_t} \right)^{-1/\gamma_1} x, & \text{if } x < \tilde{x}, \\[1ex]
        \bar{b}, & \text{if } \tilde{x} \leq x < x^*, \\[1ex]
        0, & \text{if } x \geq x^*,
    \end{cases}
    \label{eq:Optimal_Labor_Supply_main} \\[2ex]
    \pi_t^* &=
    \begin{cases}
        -\dfrac{\theta}{\sigma} \dfrac{V'_{\mathrm{int}}(x)}{V''_{\mathrm{int}}(x)}, & \text{if } x < \tilde{x}, \\[1ex]
        -\dfrac{\theta}{\sigma} \dfrac{V'_{\bar{b}}(x)}{V''_{\bar{b}}(x)}, & \text{if } \tilde{x} \leq x < x^*, \\[1ex]
        \dfrac{\theta}{\sigma \gamma_1} x, & \text{if } x \geq x^*,
    \end{cases}
    \label{Optimal_Investment_main_aligned} \\[2ex]
    \tau^* &= \inf\{ t \geq 0 : x \geq x^* \}.
    \label{Optimal_Stopping_main_aligned}
\end{align}
\noindent where $\delta_t$ represents the age-dependent force of mortality at age $t$ defined in equation \eqref{eq:Gompertz}, the functions $V'_{int}(x)$, $V'_{\bar{b}}(x)$, $V''_{int}(x)$ and $V''_{\bar{b}}(x)$ represent the first and second derivatives of the value function in the respective regions in Theorem \ref{thm:value_function}.
\end{theorem}

\begin{corollary}[Optimal Annuity Payment Rate]\label{cor:optimal_annuity_payment_rate}
Given optimal policies as in Theorem \ref{thm:optimal_policies}, with the optimal retirement wealth threshold $x^*$, the \textit{optimal annuity payment rate} $k_t^*$ is
\begin{equation}
k_t^* =
\begin{cases}
0, & \text{if } -\frac{w\bar{b}}{r} < x < x^* \quad \text{(working period)}, \\[6pt]
\phi x, & \text{if } x \geq x^* \quad \text{(full retirement period)},
\end{cases}
\end{equation}
where $\phi$ is the endogenous withdrawal rate and satisfies
\begin{equation}
\phi = \left( \frac{c_t^*}{\rho^{\frac{1}{\gamma_1}}x}\right)^{\frac{\gamma_1}{\gamma_1 -1}}\bigg|_{x \geq x^*} = k,
\end{equation}
with $c_t^*$ given by the optimal consumption policy in~\eqref{eq:optimal_consumption_policy_aligned}.
\end{corollary}

\section{Results and Discussion}\label{Numerical_Implementation}
\noindent We focus on the continuation (working period) region ($x < {x^*}$), where $V(x) > G(x)$, and the stopping (full retirement) region ($x \geq {x^*}$), where $V(x) = G(x)$. The value ${x^*}$ is the optimal retirement wealth threshold. We implement the optimal results outlined in Theorem 
\ref{thm:value_function} and Theorem \ref{thm:optimal_policies} and discuss their practical applications in retirement planning.  

\subsection{Numerical Implementation and Results}
\noindent We set the model parameters following the common setting of \cite{gerrard2012choosing}, \cite{chen2021optimal}, and \cite{gao2022optimal}, for a male retiree at age $60$ with a \textit{varying} force of mortality $\delta_t$. Unless otherwise specified, the parameters are: \( w = 10 \) (wage rate), \( \alpha = 0.2 \) (weight for consumption in the period utility function), \( r = 0.02 \) (interest rate), \( \gamma = 2 \) (risk aversion coefficient), \( \gamma_1 = 1.2 \) (risk aversion coefficient for the second part of consumption), \( \theta = 0.07 \) (market price of risk factor),  subjective discount factor \( \beta \in (0.01,0.055)\),
and \( \bar{b} = 1 \) 
(upper bound for labor income).
\begin{figure}[h!]
\centering
\includegraphics[width=0.8\textwidth]{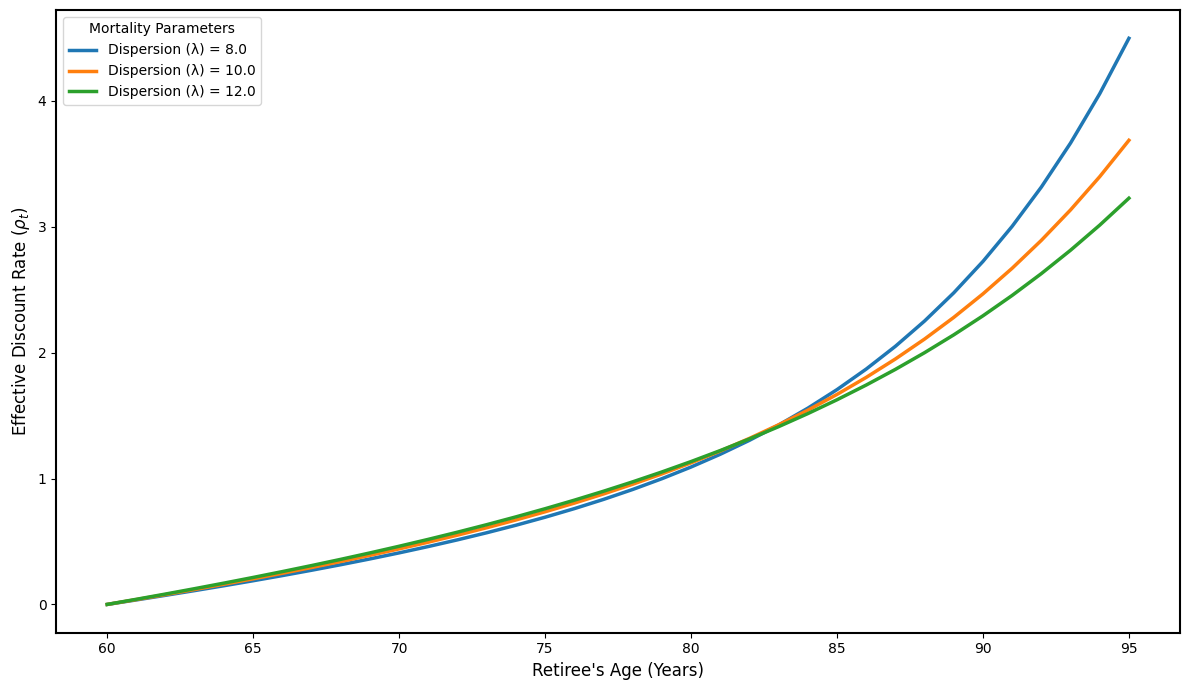}
\caption{Evolution of the effective discount rate ($\rho_t$) for an agent starting at age 60, shown for different mortality dispersion parameters ($\lambda$). Age-dependent mortality causes $\rho_t$ to rise sharply, reflecting higher mortality risk at advanced ages.}
\label{fig:rho_age}
\end{figure}

\autoref{fig:rho_age} examines the impact of the effective discount rate $\rho_t$, which reflects the age-dependent force of mortality. Compared to existing research on optimal annuitization, including studies by \cite{gerrard2012choosing} and 	\cite{gao2022optimal}, which assume a constant mortality rate, our findings demonstrate that $\rho_t$ increases with age, indicating higher mortality risk. This results in greater discounting of future annuity payments, as the probability of survival decreases. At older ages, higher $\rho_t$ values lower the present value of annuities, making them less appealing. Consequently, retirees may choose to keep more liquid assets or rely on alternative income sources, such as employment income. Conversely, younger retirees encounter lower $\rho_t$ values, leading to higher present values for annuities and more substantial incentives to annuitize.

Labor income acts as a partial hedge against longevity risk, as individuals with stable earnings might delay annuitization. Nonetheless, the influence of age-dependent mortality remains significant because even with labor income, $\rho_t$ continues to be a key factor in annuity valuation. As shown in
\autoref{fig:rho_age}, while older individuals qualify for higher annuity payouts due to shorter life expectancy, these benefits are offset by the increased discount rate $\rho_t$. Therefore, the optimal level of annuitization hinges on the interaction between mortality risk (via $\rho_t$), the retiree's age, the availability of other income sources, and liquidity needs. In summary, the effective discount rate $\rho_t$, which incorporates age-dependent mortality effects, plays a vital role in shaping optimal annuitization decisions.

There are two competing effects for an older individual. First, the agent's perspective (demand side): the agent's personal discount rate $\rho_t$ increases with age (see \autoref{fig:rho_age}). They value a dollar less next year, which reduces their desire to annuitize. Second, the insurer's perspective (supply side): insurers offer better payout rates to older people. Since an 85-year-old has a shorter life expectancy than a 65-year-old, the insurer can offer a much higher annual payment for the same lump sum. This is often referred to as \textit{mortality credits} and increases the incentive to annuitize.

\begin{figure}[h!]
\centering
\includegraphics[width=0.8\textwidth]{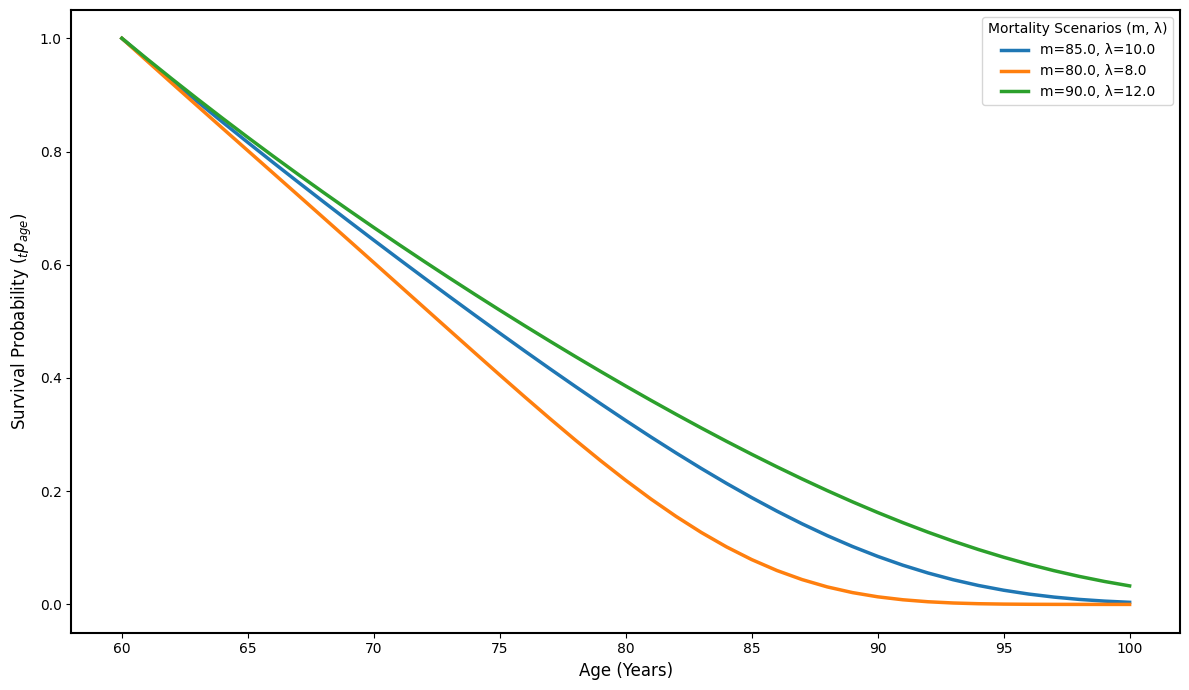}
\caption{Agent survival probability from age 60 ($\px[t]{60}$) under different mortality scenarios, defined by the modal age of death ($m$) and the dispersion parameter ($\lambda$).}
\label{fig:survival_age}
\end{figure}

In \autoref{fig:survival_age}, we observe that survival probabilities decline sharply after age 70, causing the effective discount rate, $\rho_t$, to increase. This rising discount rate reveals a critical trade-off in the decision to annuitize. While conventional wisdom suggests that rising mortality risk should increase the demand for annuities as a form of longevity insurance, our model highlights a powerful counteracting force: a high personal discount rate diminishes the subjective value of future annuity payments.

\begin{figure}[h!]
 \centering
\includegraphics[width=0.8\textwidth]{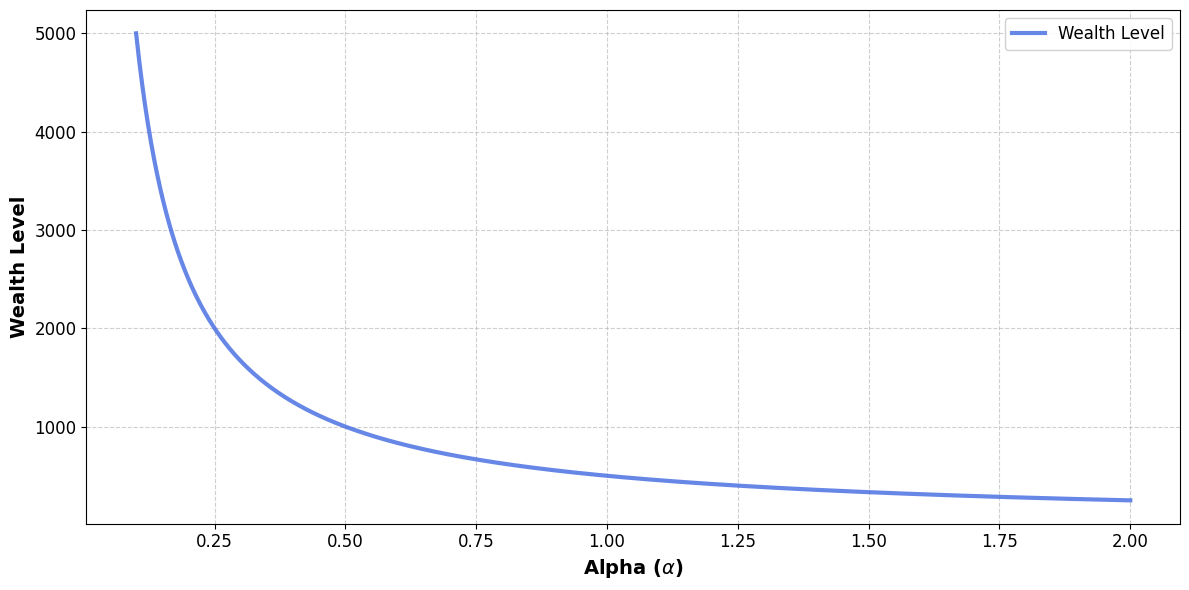} %
\caption{Wealth level and the contribution of consumption to the retiree's period utility. %
}
\label{fig:Wealth_alpha}
\end{figure}

At younger retirement ages (e.g., 60--70), the survival probability is high and $\rho_t$ is relatively low. However, the perceived need for longevity insurance is distant, so individuals may prefer to maintain liquidity and rely on other income sources. As they age into their 70s and 80s, the sharp increase in $\rho_t$ means they discount future income so heavily that even favorably priced annuities become unappealing. Essentially, the individual becomes increasingly unwilling to sacrifice present wealth for a future income stream they feel progressively less likely to receive. Therefore, our model suggests that the timing and extent of annuitization are governed by the tension between the need for longevity insurance and the declining subjective value of future income at advanced ages.

In \autoref{fig:Wealth_alpha}, we observe that the wealth level at annuitization is a decreasing function of $\alpha$, which is the weight for consumption in the period utility function. With a lower $\alpha$, consumption contributes more to the agent's utility.  The agent tends to annuitize and fully retire at a lower level of wealth.

\autoref{fig:optimal_consumption} illustrates the optimal consumption policy across two distinct regions: the continuation region ($x < x^*$), where it is optimal for the agent to continue working because wealth is below the retirement threshold, and the stopping (full retirement) region ($x \ge x^*$), where the agent's wealth has reached or exceeded the retirement threshold. Different behaviors in these wealth regions characterize the policy.

\begin{figure}[h!]
 \centering
 \includegraphics[width=1\textwidth]{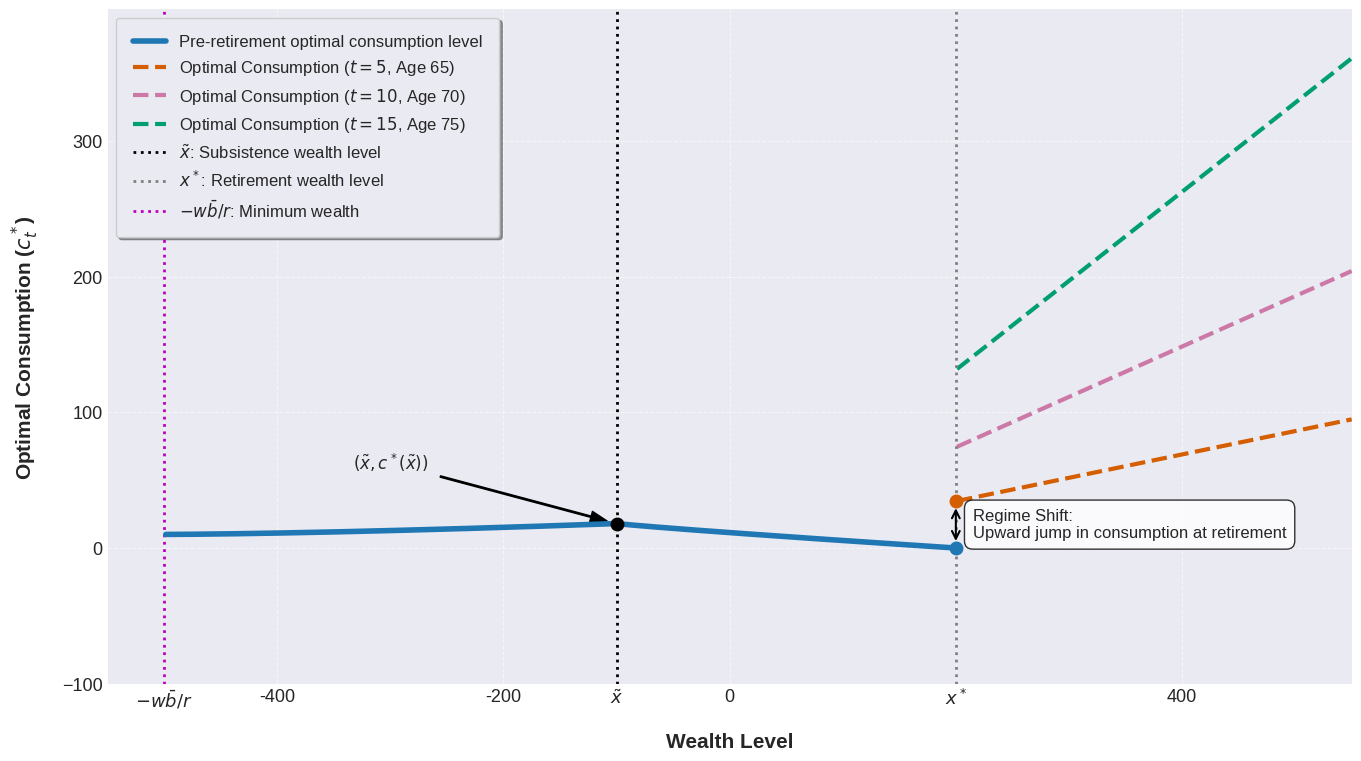} %
 \caption{Optimal consumption and wealth level. The solid blue line represents the pre-retirement optimal consumption path. The dashed lines show post-retirement consumption paths at different time horizons ($t=5, 10, 15$ years post-retirement, corresponding to ages 65, 70, and 75). The vertical dotted lines indicate key thresholds: the subsistence wealth level ($\tilde{x}$) and the retirement wealth level ($x^*$). Notably, at the point of retirement ($x^*$), there is a discrete upward jump in consumption, reflecting a regime shift in the retirees' spending behavior.}
 \label{fig:optimal_consumption}
\end{figure}

Before retirement, in \textit{pre-retirement}, as shown in \autoref{fig:optimal_consumption}, the agent's optimal consumption behavior increases with wealth before reaching the retirement wealth target. This holds for wealth levels below the subsistence wealth level $\tilde{x}$, which drops after the subsistence wealth level $\tilde{x}$. The agent prioritizes wealth accumulation to achieve their retirement goal, resulting in a relatively flat consumption curve. Upon reaching the wealth level $x^*$, the agent fully retires. This triggers an immediate and discrete upward jump in consumption. The blue linear optimal consumption function (see \autoref{fig:optimal_consumption}) represents the consumption level just before retirement, while the linear optimal consumption function at age $65$ represents the higher consumption level immediately after retiring.

In \textit{post-retirement~($x \geq x^*$)},  or \textit{full retirement}, consumption is represented by the family of linear optimal consumption functions at different ages (see \autoref{fig:optimal_consumption}). In this phase, the agent's optimal consumption behavior is significantly more sensitive to wealth, exhibiting a steeper, linear relationship. Furthermore, the consumption path shifts upward over time (from age 65 to 75), indicating that for a given level of wealth, the older agent consumes more, potentially due to a shorter remaining time horizon.

The optimal consumption function is therefore piecewise linear, with a notable discontinuity at the retirement threshold $x^*$. Contrary to a model where retirement might necessitate a drop in spending, this specification shows that achieving the threshold wealth, the agent unlocks a higher level of consumption. This is driven by the fact that the agent no longer needs to suppress consumption to save for a future goal and can begin to decumulate or spend the returns from their accumulated wealth and annuity payments.

\begin{figure}[h!]
 \centering
\includegraphics[width=.9\textwidth]{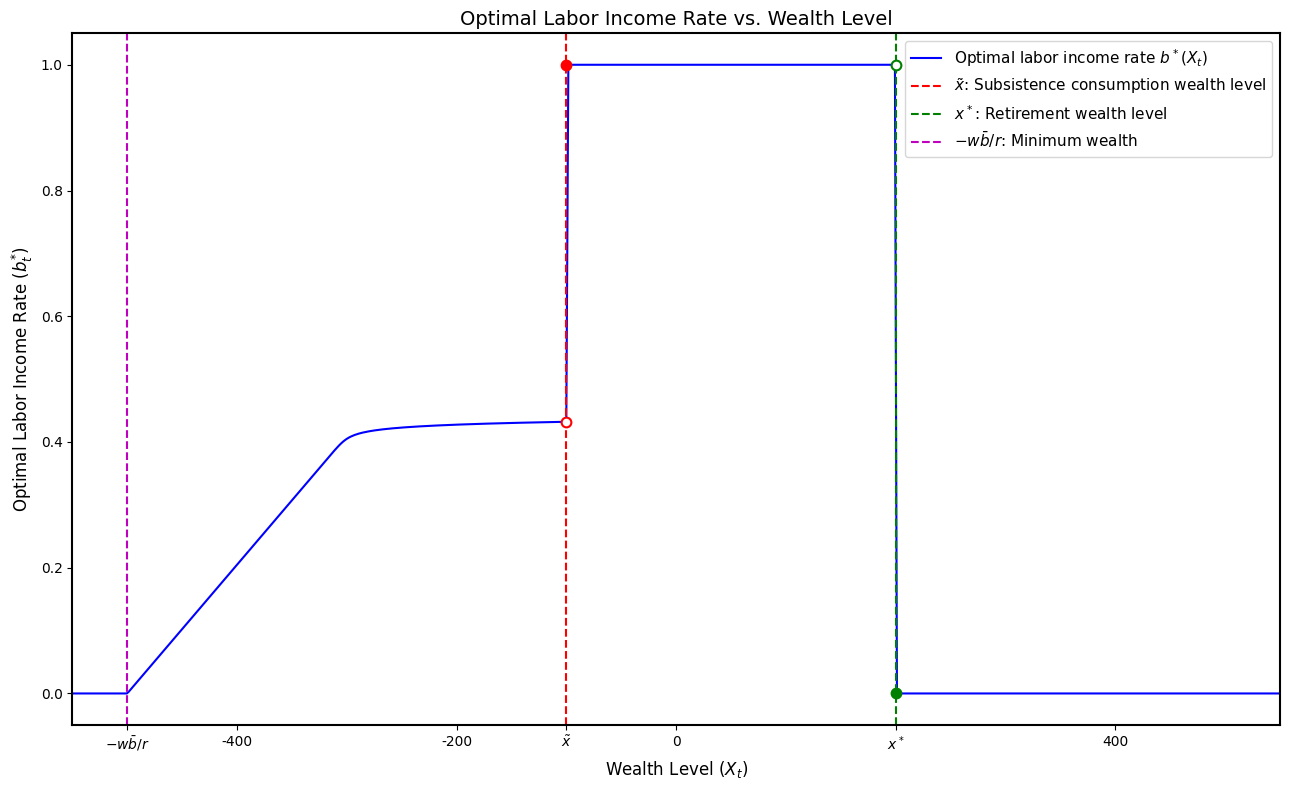}  %
 \caption{Optimal labor income and wealth level. The red dashed line represents the threshold wealth level corresponding to the labor supplied. The green dashed line represents the threshold wealth level corresponding to the optimal retirement time \( \tau \). The discontinuities at \(\tilde{x}\) and \({x^*}\) represent sudden shifts in the retiree's optimal labor supply as their wealth crosses critical thresholds. These shifts are linked to retiree decisions about annuitization (in the broader sense of securing future income streams) and the transition into a phase of reduced or no labor (retirement).}
 \label{fig:optimal_labor}
\end{figure}

In \autoref{fig:optimal_labor}, we observe that before retirement, the agent receives an increasing income from labor until the threshold wealth level \(\tilde{x}\) is reached. In the low-wealth region ($x < \tilde{x}$), the agent is actively engaged in labor or borrowing to supplement their resources and maintain a basic consumption level. This is a phase of significant reliance on labor income or debt.

As wealth reaches \(\tilde{x}\), there is a discrete jump upwards in the optimal labor income rate to \(\bar{b}\). This marks a change in strategy. At this wealth level, the agent is more secure and aims for a higher consumption level associated with a full retirement lifestyle. The jump reflects a temporary increase in labor supply to reach a desired savings level and align income with pre-retirement spending habits. Reaching \(\tilde{x}\) triggers a shift toward accumulating wealth for future annuitization, either formally through financial products or informally through savings drawdown. In the intermediate wealth range, the agent maintains a constant labor income rate of \(\bar{b}\).

Once an agent's wealth reaches the retirement threshold, $x^*$, they exit the labor force. This stage marks the transition to \textit{full retirement}, characterized by the complete cessation of labor income. From this point forward, the agent relies solely on their accumulated assets to support their desired consumption through accumulated wealth and annuity payments. 

\begin{figure}[h!]
\centering
\includegraphics[width=.8\textwidth]{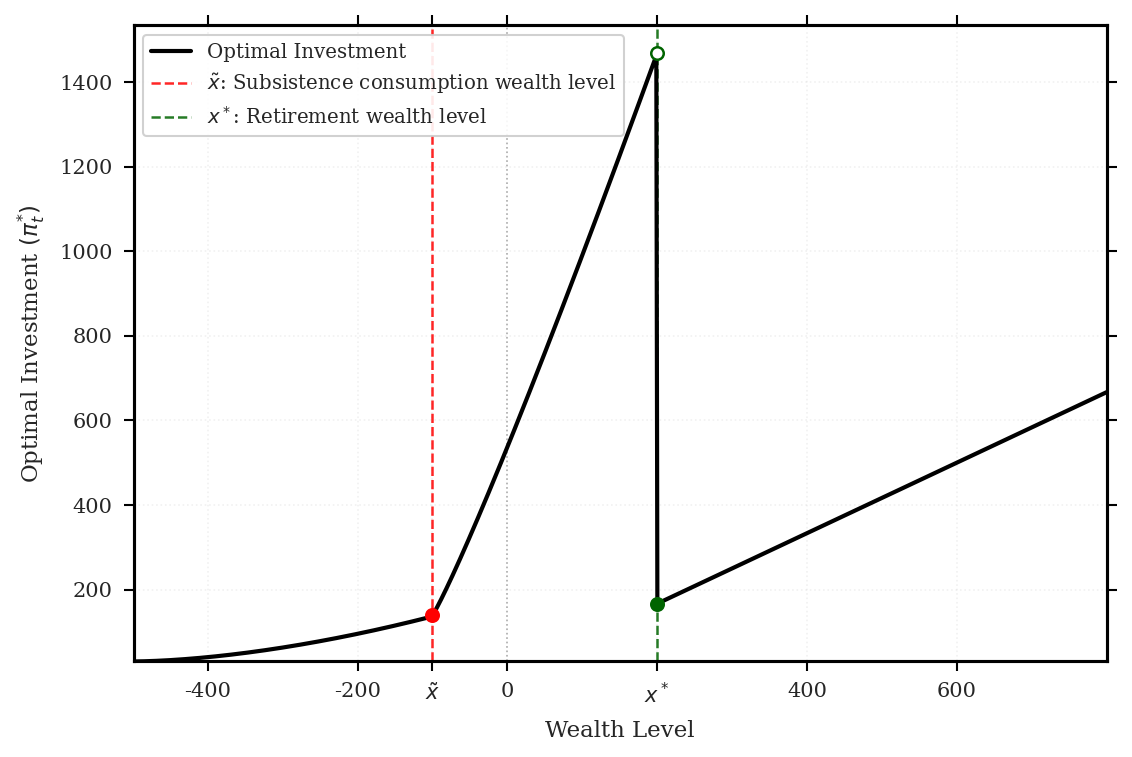} %
\caption{Optimal investment and wealth level. The red dashed line represents the threshold wealth level corresponding to the labor supply. The green dashed line represents the threshold wealth level corresponding to the optimal retirement time \( \tau \). The discontinuity (a jump) in the optimal investment strategy at \( x = \tilde{x} \) signifies a sudden change in the proportion of wealth allocated to the risky asset as the retiree's wealth crosses the threshold \( \tilde{x} \). The discontinuity at $x = {x^*}$ represents a shift from a constant investment amount (as a proportion of some base consumption level) to an investment amount that is linearly increasing with wealth.}
\label{fig:optimal_investment}
\end{figure}
In analyzing the optimal investment strategy shown in \autoref{fig:optimal_investment}, a key feature is the discontinuity at the retirement wealth threshold, $x^*$. The jump is the fundamental and instantaneous change in the agent's investment strategy, driven by the economics of labour income. Before full retirement, an agent's total wealth consists of both financial wealth and labour income, with the latter representing the present value of all future labor income. This labor income acts as a large, stable, bond-like asset. To achieve an optimal risk balance across their total portfolio, the agent can aggressively invest their financial wealth in risky assets and safe portion represented by their future income. This explains the high level of investment just before the retirement threshold.

At the moment of full retirement, when wealth reaches $x^*$, labour income vanishes. The agent's total wealth is now composed entirely of their financial portfolio. To avoid being over-exposed to market risk, the agent must immediately rebalance by drastically reducing their allocation to risky assets. The sharp downward jump in the optimal investment, $\pi_t^*$, represents this strategic de-risking. The optimal policy we derived from our model (see \eqref{Optimal_Investment_main_aligned}) captures this by switching from a complex value function that accounts for labor income to the classic Merton portfolio rule for a fully retired agent with no human capital.

\autoref{fig:optimal_investment} shows the optimal proportion of risky investment. We observe a substantial decline in \textit{investment} by the agent after full retirement. As reported in \cite{ferrari2023optimal}, this result implies that agents start saving for retirement as a de-risking strategy as people tend to withdraw their investment in risky assets upon retirement to hedge the risk of unemployment.  \autoref{fig:optimal_investment} also shows the optimal proportion of risky investment. We observe a significant decline in \textit{investment} by the agent after full retirement. As reported in \cite{ferrari2023optimal}, this result suggests that agents begin saving for retirement as a de-risking strategy, since people tend to withdraw their investments in risky assets upon retirement to hedge against unemployment risk.

The numerical results presented in Figures \ref{fig:optimal_consumption} and \ref{fig:optimal_investment} are consistent with existing research on optimal annuitization (see, e.g., \cite{coile2009household, ferrari2023optimal}). A key finding of our model is that the effective discount rate, $\rho_t$, increases with age, directly reflecting a higher mortality risk. Consequently, at older ages, these higher $\rho_t$ values lower the subjective present value of annuities, making them a less appealing financial instrument for wealth management.

\begin{figure}[h!]
\centering
\includegraphics[width=0.8\textwidth]{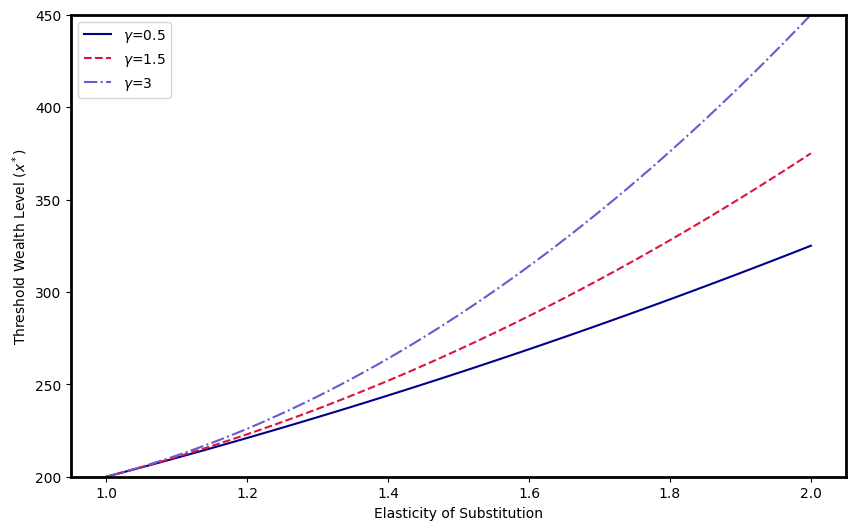} %
\caption{Retirement wealth level and elasticity of substitution. %
}
\label{fig:threshholdwealth_elecitcity_sub}
\end{figure}

In \autoref{fig:threshholdwealth_elecitcity_sub}, we observe that the critical wealth level ${x^*}$ increases with the elasticity of substitution between consumption and leisure. An agent with a higher elasticity of substitution between consumption and leisure can consume more and achieve greater utility than an agent with a lower elasticity of substitution. Consequently, the former has a stronger incentive to continue working, leading to a higher threshold retirement wealth level and a tendency to delay full retirement.

\begin{figure}[h!]
\centering
\includegraphics[width=0.8\textwidth]{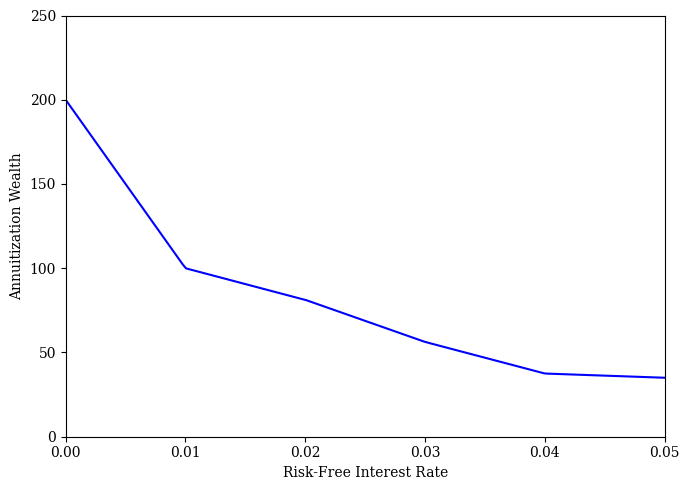} %
\caption{Sensitivity analysis of threshold wealth for annuitization for changes in the risk-free interest rate.}
\label{fig:wealth_risk}
\end{figure}
\autoref{fig:wealth_risk} illustrates the relationship between the risk-free interest rate and the wealth required for annuitization in retirement planning. We observe that in a low-interest-rate environment, retirees may need to allocate a larger portion of their retirement savings to an annuity to achieve their desired income level compared to a high-interest-rate environment. This relationship highlights the importance of considering broader economic conditions, particularly interest rate trends, when making annuitization decisions as part of retirement planning. \autoref{fig:wealth_risk} also shows that the cost of securing a guaranteed income stream through an annuity is inversely related to the prevailing risk-free interest rate, as high interest rates reduce the initial investment required for annuitization.

\autoref{fig:annuity_payment_rate} illustrates an agent's optimal annuity payment rate, \( k_t^* \), which characterizes the individual's annuitization strategy conditional on their wealth, \( x \). During the working period (\( x < x^* \)), the \textit{optimal annuity rate is zero}, as the individual derives income from labor and postpones annuitization. This reflects the strategy of preserving liquid wealth for when labor income ceases. Upon reaching the retirement wealth threshold (\( x \ge x^* \)), the optimal strategy is not to withdraw a constant percentage but to engage in a \textit{dynamic, phased annuitization}. The process begins with a relatively small annuity rate at the onset of retirement. As the agent ages, the optimal rate, \( k_t^* \), \textit{steadily increases}. This indicates that the agent should progressively allocate more wealth towards securing a guaranteed lifetime income, allowing them to balance the need for portfolio growth in early retirement with the increasing importance of mitigating longevity risk in later life.

\subsection{Optimal Policies Analysis}\label{Optimal_Policies_Analysis}
Retirement planning represents a critical phase in an individual's life, demanding careful management of accumulated wealth to ensure a desired level of consumption and financial security over a potentially extended period, typically spanning from age 60 to 90. This section focuses on the analysis of theoretically derived optimal policies in Theorem \ref{thm:optimal_policies} for three key decision variables faced by agents: consumption, labor income rate, and investment. The mathematical expressions for these optimal policies define how a agent should ideally behave to these variables to maximize their well-being. We provide a detailed economic interpretation of the various policy regimes and the underlying economic parameters that shape them. 

\begin{figure}[h!]
\centering
\includegraphics[width=0.8\textwidth]{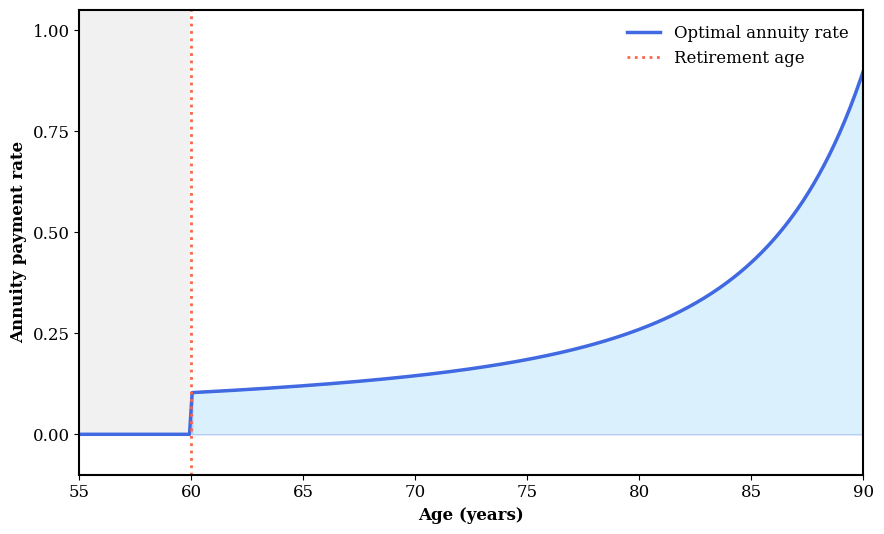} %
\caption{Optimal annuity payment rate with age.}
\label{fig:annuity_payment_rate}
\end{figure}

Within the dynamic optimization framework of our model, the agent's total wealth at any given time $t$, denoted by $x$, functions as the key state variable. This variable captures the agent's current financial standing and determines the regime for each optimal policy. The agent's decisions, or control variables, include the optimal consumption rate ($c_t^*$), which is the amount of wealth they choose to spend on goods and services at time $t$. The optimal labor income rate ($b_t^*$) indicates the rate at which the agent earns income from work, if they opt to do so at time $t$. The optimal investment strategy ($\pi_t^*$) specifically refers to the proportion of wealth allocated to a risky asset at time $t$. The agent selects these control variables at each moment to maximize their overall utility, as outlined by the economic model in Section 
\ref{market_model}. 

A key feature of the derived optimal policies is their piecewise nature. Each policy \( ( c_t^{*}, b_t^{*}, \pi_t^{*} ) \) adopts different functional forms depending on the current wealth level $x$. Changes in these policies are triggered by specific wealth thresholds: $w \overline{b}/r$, $	\tilde{x},$ and ${x^*}$. These thresholds represent critical points at which the agent’s optimal behavior undergoes a qualitative shift. The lower threshold, \( w \overline{b}/r\), represents the present value of the agent's future labor income and a boundary linked to borrowing constraints or a minimum wealth needed to sustain a certain level of consumption. The other two thresholds $\tilde{x}$ and ${x^*}$, are described in Propositions  \ref{prop:threshold_x_tilde} and \ref{prop:optimal_threshold} represent the agent's optimal behavior shift linked to labour supply and full retirement wealth threshold, respectively. They are significant milestones in the agent's financial journey. These milestones may signal the point at which they can transition to a higher standard of living or cease relying entirely on labor income. The presence of these distinct wealth regimes highlights the adaptive nature of the optimal retirement strategy, where decisions depend on the agent's evolving financial situation.

Figures \ref{fig:optimal_consumption}, \ref{fig:optimal_labor}, and \ref{fig:optimal_investment} illustrate that the optimal policies are \textit{piecewise functions of wealth}. The first behavioral regime begins at the solvency threshold, which, given the parameters $w=10$, $\bar{b}=1$, and $r=0.02$, is calculated as $-w\bar{b}/r = -500$. The figures show how the functional form for consumption, labor, and investment changes as wealth crosses this and other critical thresholds.

\subsubsection{Analysis of Optimal Consumption Policy}%
\noindent The optimal consumption policy is structured across three distinct wealth regimes, each dictating how spending responds to changes in accumulated wealth ($x$). In the low-wealth regime, where $- \frac{w\bar{b}}{r} < x < \tilde{x}$, consumption is not fixed but varies with current wealth. This variation is determined by the marginal value of additional wealth, $V'_{\mathrm{int}}(x)$, ensuring that expenditures, including those for essential needs, are precisely aligned with the agent's financial position. In the intermediate regime, defined by $\tilde{x} \leq x < x^*$, optimal consumption remains flexible and continues to adjust as assets increase. Here, the marginal value of wealth, $V'_{\bar{b}}(x)$, governs spending, allowing the standard of living to improve progressively with financial growth. In the high-wealth regime, where $x \geq x^*$, consumption becomes directly proportional to wealth.  At this stage, the agent consumes a fixed proportion of assets annually, indicating that basic consumption needs are met and spending scales with asset size. This transition from wealth-sensitive consumption in lower regimes to a proportional rule at high wealth reflects a systematic approach to retirement financial management, where spending strategies shift from detailed adjustments to efficient resource utilization as financial security increases.

\subsubsection{Analysis of Optimal Labor Income Rate}%
\noindent The optimal labor income rate, $b_t^*$, demonstrates piecewise behavior as a function of wealth, $x$. In the low wealth regime, defined by $x < \tilde{x}$, $b_t^*$ is positive and increases linearly with wealth, as described by $b_t^* = \left(\frac{1 - \alpha}{w\alpha}\right) k^{\frac{\gamma_1-1}{\gamma_1}} \rho_{t}^{\frac{1}{\gamma_1}} x$. Under limited wealth, the agent is required to engage in labor to supplement income, and the labor supply adjusts in direct proportion to the wealth level within this regime.

In the intermediate wealth range, where $\tilde{x} \leq x < x^*$, the optimal labor income rate remains constant at its maximum value, $\bar{b}$. This regime is characterized by continued labor participation to facilitate further wealth accumulation or to sustain elevated consumption levels. When wealth enters the high regime, defined by $x \geq x^*$, the optimal labor income rate reduces to zero. At this stage, the agent no longer requires labor income and allocates all available time to leisure. This structure establishes a non-monotonic relationship: labor supply increases with wealth in the low regime, remains at its maximum in the intermediate regime, and ceases entirely in the high wealth regime.

\subsubsection{Analysis of Optimal Investment Strategy}%
\noindent The optimal investment strategy, $\pi_t^*$, is a piecewise function of the agent's wealth, $x$, with the policy shifting at critical wealth thresholds. The complete policy is defined as in equation \eqref{Optimal_Investment_main_aligned}. In the continuation (working period)  regimes ($x < x^*$), the amount invested in the risky asset is determined by the agent's absolute risk tolerance (the term $-V'/V''$) and the market price of risk ($\theta/\sigma$). This reflects a dynamic strategy where the agent's willingness to accept risk evolves with their wealth.

Upon entering retirement ($x \geq x^*$), the strategy simplifies to the classical Merton’s problem (see \cite{koo2013optimal}) where the agent invests a \textit{constant proportion} of their wealth, $\theta/\sigma \gamma_1$, in the risky asset. This shift represents a transition from the complex risk adjustments required during the wealth accumulation phase to a more stable asset allocation, focusing on wealth preservation and long-term growth in full retirement.

\subsubsection{Role of Parameters}
\noindent The model parameters quantitatively determine the locations of wealth thresholds and the magnitudes of optimal policies. The optimal investment strategy exhibits significant sensitivity to the agent's risk preferences and current market conditions. An increase in the risk aversion coefficient, $\gamma_1$, decreases the proportion of wealth allocated to the risky asset, as reflected in the term $\theta/(\sigma\gamma_1)$. Investment is directly related to the market price of risk, represented by the Sharpe ratio, $\theta/\sigma$. Additional parameters characterize the agent's constraints and available opportunities. An increase in the risk-free interest rate, $r$, shifts the solvency boundary, $-w\bar{b}/r$, thereby affecting policy regimes. Moreover, a higher wage rate, $w$, amplifies the effect of employment, since labor income, $w \cdot b_t^*$, is directly proportional to the wage.

To summarize the numerical results,  Section \ref{Optimal_Policies_Analysis} analyzes the optimal policies for an \textit{agent's} consumption, labor, and investment. The results reveal a sophisticated framework where decisions are piecewise functions of the agent's current wealth, defining distinct behavioral phases separated by critical thresholds. As illustrated in Figures \ref{fig:optimal_consumption}, \ref{fig:optimal_labor}, and \ref{fig:optimal_investment}, these strategies are consistent with findings in the existing literature (see, e.g., Section 6.2 of \cite{ferrari2023optimal}). The analysis also demonstrates that the model's underlying economic parameters, such as risk aversion and interest rates, are critical in shaping these thresholds and policy values. While this theoretical framework provides valuable insights, real-world retirement decisions often involve additional factors that necessitate more comprehensive planning tools.

\section{Conclusion}\label{Conclusion_Recommendations}
\noindent In this paper, we address an optimal annuitization problem with labor income, considering an age-dependent force of mortality. We derive closed-form solutions for the value function and the optimal joint strategy for investment, consumption, and labor supply. This work highlights practical applications of our findings in retirement planning, offering valuable insights for individuals and policymakers. Key results show that an \textit{agent's} strategies change as their wealth surpasses certain thresholds, with the availability of labor income strongly motivating wealth accumulation. Our results also indicate that in low-interest environments, a larger share of savings must be allocated to an annuity to meet income goals.

Furthermore, our findings demonstrate that the effective discount rate increases with age, reflecting a higher mortality risk. This creates \textit{two competing effects} for an older \textit{agent}. First, as our model shows, the \textit{agent's} personal discount rate, $\rho_t$, increases with age (see \autoref{fig:rho_age}). They value a dollar less next year, which reduces their desire to annuitize. Second, from the insurer's perspective, older \textit{agents} are offered better payout rates for the same lump sum. This effect, often referred to as \textit{mortality credits}, increases the incentive to annuitize. Our model's key finding is that the disincentive from the high personal discount rate can outweigh the incentive from these favorable mortality credits. Therefore, the optimal level of annuitization depends critically on the interplay of mortality risk, the \textit{agent's} age, alternative income sources, and liquidity requirements.

This research has certain limitations. The model abstracts from several institutional and behavioral features that may influence annuitization decisions, such as age-dependent labor income, fixed retirement ages, and the presence of employer or public pension benefits. Extending the analysis to incorporate these elements, or to allow for partial and sequential annuitization as mortality risk evolves, would be a natural avenue for future research. Nonetheless, even within this stylized framework, introducing an age-dependent force of mortality reveals nuanced and sometimes counterintuitive interactions between longevity risk, consumption smoothing, and annuity demand. These insights underscore the theoretical importance of age-varying mortality in shaping optimal retirement and annuitization strategies and provide a foundation for future empirical and policy-oriented studies.

\vspace{1em}

{\bf Acknowledgements}: We acknowledge the support of the Natural Sciences and Engineering Research Council of Canada (NSERC), funding reference number RGPIN-2021-04112.  Nous remercions le Conseil de recherches en sciences naturelles et en g\'enie du Canada (CRSNG) de son soutien, num\'ero de r\'ef\'erence RGPIN-2021-04112.

\bibliographystyle{plainnat}
\bibliography{reference2.bib}

\appendix
\section{Appendix: Proofs of Theorems}
This appendix contains the proofs of theorems that would disrupt the flow of the paper.
\numberwithin{equation}{subsection}

\subsection{Proof of Theorem \ref{thm:value_function} (Value Function):}\label{Proof_of_Theorem_1:_Value_Function}
  We consider the pre-retirement problem ($t<\tau$). The value function $V(x)$ satisfies the HJB equation \eqref{eq:HJB_cont_region_rho_appendix} in the continuation region $x<x^*$, where $V(x)>G(x)=\frac{k^{1-\gamma_1}}{\rho_t(1-\gamma_1)}x^{1-\gamma_1}$. We analyze the problem by dividing the continuation region based on the optimal labor supply control $b_t$.

\vspace{0.5em}
\noindent
Let $V(x)$ be the stationary value function satisfying the HJB equation in equation \eqref{eq:HJB_cont_region_rho_appendix}, 
\begin{equation}\label{eq:HJB_pre_proof_detail}
\rho_{t} V(x) = \sup_{c\ge0, 0 \leq b \leq \bar{b}, \pi} \left[ u_1(c, b) + \left( r x + \pi (\mu - r) - c + w b \right) V'(x) + \frac{1}{2} \sigma^2 \pi^2 V''(x) \right].
\end{equation}
where $\rho_t$ is the effective discount rate $\rho_t$ in equation \eqref{eq:age-dependent_force_of_mortality}, which incorporates age-dependent force of mortality at age $t$ defined in equation \eqref{eq:Gompertz}.

\textbf{Case I:} $-r w \bar{b} < x < \tilde{x}$ (Interior Labor Supply $0 \le b^* < \bar{b}$). 
\vspace{0.5em}
In this region, the optimal controls are determined by the first-order conditions (FOCs) derived in  \autoref{thm:optimal_policies}
\begin{align}
\pi^*(x) &= -\frac{\theta}{\sigma} \frac{V'(x)}{V''(x)}, \label{eq:pi_star_detail} \\
V'(x) &= \frac{\partial u_1(c^*, b^*)}{\partial c}, \label{eq:c_star_detail} \\
\frac{\partial u_1(c^*, b^*)}{\partial b} &= -w V'(x). \label{eq:b_star_detail}
\end{align}
Using the Cobb-Douglas utility $u_1(c,b) = \frac{1}{1-\gamma}(c^\alpha b^{1-\alpha})^{1-\gamma}$, conditions \eqref{eq:c_star_detail} and \eqref{eq:b_star_detail} imply the relationship $c^* = -\frac{w\alpha}{1-\alpha} b^*$, or equivalently $b^* = -\frac{1-\alpha}{w\alpha} c^*$.
Substituting these optimal controls into the HJB equation \eqref{eq:HJB_pre_proof_detail}, we obtain the maximized HJB equation for $V(x)=V_{\text{int}}(x)$. We express the utility and drift terms using $c^*$:
$b^* = -\frac{w\alpha}{1-\alpha}c^*$, 
$u_1(c^*,b^*) = \frac{h}{1-\gamma}(c^*)^{1-\gamma}$, where $h=\alpha^{\alpha(1-\gamma)}(-\frac{w\alpha}{1-\alpha})^{(1-\alpha)(1-\gamma)}$.
$-c^*+wb^* = -c^*+w(-\frac{1-\alpha}{w\alpha}c^*) = -\frac{1}{\alpha}c^*$.

Substituting $\pi^*$ using \eqref{eq:pi_star_detail}, the maximized HJB becomes
\begin{equation}\label{eq:HJB_substituted_detail}
\rho_{t} V_{\text{int}}(x) = \frac{h}{1-\gamma}(c^*)^{1-\gamma} + \left( rx - \frac{1}{\alpha}c^* \right)V_{\text{int}}'(x) - \frac{1}{2} \theta^2 \frac{(V_{\text{int}}'(x))^2}{V_{\text{int}}''(x)}.
\end{equation}
Now, assume optimal consumption $c^*=C_{\text{int}}(x)$ is invertible with inverse $x=X_{\text{int}}(c)$. We need to express $V_{\text{int}}'(x)$ and $V_{\text{int}}''(x)$ in terms of $c=c^*$. From FOC \eqref{eq:c_star_detail}, we have
\[
V_{\text{int}}'(x)=\frac{\partial u_1(c,b^*(c))}{\partial c} = K_v c^{-\gamma}, \quad \text{where } K_v = \left(-\frac{w\alpha}{1-\alpha}\right)^{(1-\alpha)(1-\gamma)}.
\]
(Note: $K_v$ involves potentially non-standard terms if $w>0$).
The second derivative is
\[
V_{\text{int}}''(x) = \frac{d(K_v c^{-\gamma})}{dx} = \frac{d(K_v c^{-\gamma})}{dc} \frac{dc}{dx} = (-\gamma K_v c^{-\gamma-1})\frac{1}{X_{\text{int}}'(c)}.
\]
Substituting $V_{\text{int}}'$ and $V_{\text{int}}''$ into the maximized HJB \eqref{eq:HJB_substituted_detail}:
\begin{align*}
\rho_t V_{\text{int}}(X_{\text{int}}(c)) &= \frac{h}{1-\gamma} c^{1-\gamma} + \left(rX_{\text{int}}(c) - \frac{1}{\alpha}c \right)K_v c^{-\gamma} - \frac{1}{2}\theta^2 \frac{(K_v c^{-\gamma})^2}{-\gamma K_v c^{-\gamma-1}/X_{\text{int}}'(c)} \\
\rho_t V_{\text{int}}(X_{\text{int}}(c)) &= \frac{h}{1-\gamma}c^{1-\gamma} + K_v r X_{\text{int}}(c) c^{-\gamma} - \frac{K_v}{\alpha} c^{1-\gamma} + \frac{1}{2\gamma}\theta^2 K_v c^{1-\gamma}X_{\text{int}}'(c).
\end{align*}
To obtain the ODE for $X_{\text{int}}(c)$, we differentiate this equation with respect to $c$. We use the relation $\frac{d}{dc}V_{\text{int}}(X_{\text{int}}(c)) = V_{\text{int}}'(X_{\text{int}}(c))\cdot X_{\text{int}}'(c) = (K_v c^{-\gamma})X_{\text{int}}'(c)$.

\[
\begin{split}
\rho_t K_v c^{-\gamma} X_{\text{int}}'(c) ={}& hc^{-\gamma} + K_v r X_{\text{int}}'(c)c^{-\gamma} - K_v r \gamma X_{\text{int}}(c) c^{-\gamma-1} - \frac{K_v(1-\gamma)}{\alpha}c^{-\gamma} \\
& + \frac{\theta^2 K_v (1-\gamma)}{2\gamma}c^{-\gamma} X_{\text{int}}'(c) + \frac{\theta^2 K_v c^{1-\gamma}}{2\gamma}X_{\text{int}}''(c)
\end{split}
\]
Multiplying by $c^{\gamma+1}/K_v$ and rearranging terms yields a second-order linear ODE for $X_{\text{int}}(c)$:
\begin{equation}\label{eq:second_order_ODE_detail}
\mathcal{A}_{\text{int}} c^2 X''_{\text{int}}(c) + \mathcal{B}_{\text{int}} c X'_{\text{int}}(c) + \mathcal{C}_{\text{int}} X_{\text{int}}(c) = \mathcal{D}_{\text{int}} c + \mathcal{E}_{\text{int}},
\end{equation}
where $\mathcal{A}_{\text{int}}, \mathcal{B}_{\text{int}}, \mathcal{C}_{\text{int}}, \mathcal{D}_{\text{int}}, \mathcal{E}_{\text{int}}$ are constants depending on the model parameters $(\rho,r,\theta,\gamma,\alpha,h,K_v)$. Specifically: 
$\mathcal{A}_{\text{int}}=\frac{1}{2\gamma}\theta^2$, 
$\mathcal{B}_{\text{int}}=\rho_t - r - \frac{1}{2\gamma}\theta^2(1-\gamma)$, 
$\mathcal{C}_{\text{int}}=r\gamma$, 
$\mathcal{D}_{\text{int}}=\frac{1-\gamma}{\alpha} + \frac{h}{K_v}$, 
$\mathcal{E}_{\text{int}}=0$.
The general solution to \eqref{eq:second_order_ODE_detail} is $X_{\text{int}}(c)=X_h(c)+X_p(c)$. The homogeneous solution is $X_h(c)=A_1 c^{m'_1} + A_2 c^{m'_2}$ where $m'_1,m'_2$ are roots of the characteristic equation $\mathcal{A}_{\text{int}}m(m-1)+\mathcal{B}_{\text{int}}m+\mathcal{C}_{\text{int}}=0$. The particular solution $X_p(c)$ is typically found to be proportional to $c$. Assuming $m'_1>0$ and $m'_2<0$, boundary conditions usually imply $A_1=0$. Thus, the solution takes the form:
\begin{equation}\label{eq:Xc_solution_detail}
X_{\text{int}}(c, A_2) = A_2 c^{m'_2} + \text{const} \cdot c.
\end{equation}
Substituting this solution structure back into the expression for $\rho_t V_{\text{int}}(X_{\text{int}}(c))$ gives the functional form of $V_{\text{int}}(x)$ which depends on $A_2$ and involves terms related to $c=C_{\text{int}}(x)$. This corresponds to the first case in Theorem \eqref{thm:value_function}.

\vspace{1em}
\noindent
\textbf{Case II:} $\tilde{x} \le x < x^*$ (Corner Labor Supply $b=\bar{b}$)
\vspace{0.5em}
Here $b^*=\bar{b}$. The HJB equation is
\begin{equation}\label{eq:HJB_post_detail}
\rho_{t} V_{\bar{b}}(x) = \max_{(\pi, c)} \left[ u_1(c, \bar{b}) + \left( r x + \pi (\mu - r) - c + w \bar{b} \right) V_{\bar{b}}'(x) + \frac{1}{2} \sigma^2 \pi^2 V_{\bar{b}}''(x) \right].
\end{equation}
The FOCs are
\begin{align}
\pi^*(x) &= -\frac{\theta}{\sigma} \frac{V_{\bar{b}}'(x)}{V_{\bar{b}}''(x)}, \label{eq:pi_star_post_detail} \\
V_{\bar{b}}'(x) &= \frac{\partial u_1(c^*, \bar{b})}{\partial c}. \label{eq:c_star_post_detail}
\end{align}
Substituting $\pi^*$ into \eqref{eq:HJB_post_detail} yields the maximized HJB:
\begin{equation}\label{eq:HJB_substituted_post_detail}
\rho_{t} V_{\bar{b}}(x) = u_1(c^*, \bar{b}) + \left( r x - c^* + w \bar{b} \right)V_{\bar{b}}'(x) - \frac{1}{2} \theta^2 \frac{(V_{\bar{b}}'(x))^2}{V_{\bar{b}}''(x)}.
\end{equation}
Assume $c^*=C_{\bar{b}}(x)$ and $x=X_{\bar{b}}(c)$. From FOC \eqref{eq:c_star_post_detail}
\[
V_{\bar{b}}'(x) = \frac{\partial u_1(c,\bar{b})}{\partial c} = K'_{\bar{b}} c^{-\gamma_1}, \quad \text{where } K'_{\bar{b}}=(\bar{b}^{1-\alpha})^{1-\gamma} \text{ and } \gamma_1=1-\alpha(1-\gamma).
\]
Then $V_{\bar{b}}''(x) = -\gamma_1 K'_{\bar{b}} c^{-\gamma_1-1}/X_{\bar{b}}'(c)$.
Substituting $V_{\bar{b}}'$ and $V_{\bar{b}}''$ into \eqref{eq:HJB_substituted_post_detail}
\begin{align*}
\rho_t V_{\bar{b}}(X_{\bar{b}}(c)) &= u_1(c,\bar{b}) + (rX_{\bar{b}}(c) - c + w\bar{b}) K'_{\bar{b}} c^{-\gamma_1} - \frac{1}{2}\theta^2 \frac{(K'_{\bar{b}}c^{-\gamma_1})^2}{-\gamma_1 K'_{\bar{b}} c^{-\gamma_1-1}/X_{\bar{b}}'(c)}. \\
\rho_t V_{\bar{b}}(X_{\bar{b}}(c)) &= u_1(c,\bar{b}) + K'_{\bar{b}} r X_{\bar{b}}(c)c^{-\gamma_1} - K'_{\bar{b}}c^{1-\gamma_1} + K'_{\bar{b}}w\bar{b}c^{-\gamma_1} + \frac{1}{2\gamma_1}\theta^2 K'_{\bar{b}} c^{1-\gamma_1} X_{\bar{b}}'(c).
\end{align*}
Differentiating with respect to $c$ (using $\frac{d}{dc}V_{\bar{b}}(X_{\bar{b}}(c)) = V_{\bar{b}}'(X_{\bar{b}}(c))X_{\bar{b}}'(c) = K'_{\bar{b}}c^{-\gamma_1} X_{\bar{b}}'(c)$) leads to a second-order linear ODE for $X_{\bar{b}}(c)$
\begin{equation}\label{eq:second_order_ODE_post_detail}
\mathcal{A}_{\bar{b}} c^2 X''_{\bar{b}}(c) + \mathcal{B}_{\bar{b}} c X'_{\bar{b}}(c) + \mathcal{C}_{\bar{b}} X_{\bar{b}}(c) = \mathcal{D}_{\bar{b}} c + \mathcal{E}_{\bar{b}},
\end{equation}
where $\mathcal{A}_{\bar{b}}, \mathcal{B}_{\bar{b}}, \mathcal{C}_{\bar{b}}, \mathcal{D}_{\bar{b}}, \mathcal{E}_{\bar{b}}$ are constants depending on model parameters and $\bar{b}$.
The general solution is
\begin{equation}\label{eq:Xc_post_detail}
X_{\bar{b}}(c, B_1, B_2) = B_1 c^{m''_1} + B_2 c^{m''_2} + X_{p, \bar{b}}(c),
\end{equation}
where $m''_1, m''_2$ are roots of the characteristic equation for \eqref{eq:second_order_ODE_post_detail} and $X_{p, \bar{b}}(c)$ is a particular solution. Substituting back yields $V_{\bar{b}}(x)$ involving $B_1, B_2$, corresponding to the second case in Theorem \eqref{thm:value_function}.

\vspace{1em}
\noindent
\textbf{Retirement Region:} $x \ge x^*$
\vspace{0.5em}
Here $V(x)=G(x)=\frac{k^{1-\gamma_1}}{\rho_t(1-\gamma_1)}x^{1-\gamma_1}$.

\vspace{1em}
\noindent
Boundary Conditions and Constants $A_2,B_1,B_2$ and thresholds $\tilde{x},x^*$ are determined by imposing $C^2$ continuity conditions at $x=\tilde{x}$ and $x=x^*$, as detailed in Appendix \eqref{Appendix_constants}. Solving the system derived from these conditions uniquely determines the constants and thresholds.  For details on the constants $\tilde{c}$, $B_1$, $x^*$, $B_2$, $A_2$, and $\tilde{x}$, see Appendicies \ref{Proof_of_Derivation_of_tilde{c}_},
\ref{Proof_of_Derivation_of_B_1_}, \ref{Proof_of_Derivation_of_B_2_}, \ref{app:derivation_A2} and \ref{Proof_of_Derivation_of_tilde{c}_}.
Thus, the proof is completed.
\subsection{Proof of Theorem \ref{thm:optimal_policies} (Optimal Policies)}
\label{Proof_of_Theorem_3_Optimal_Policies_main_final}
The optimal policies $(\pi^*,c^*,b^*,\tau^*)$ stated in Theorem \ref{thm:optimal_policies} are obtained by applying the First-Order Conditions (FOCs) to the value function $V(x)$ determined in Theorem \ref{thm:value_function}, considering the different regions defined by the thresholds $\tilde{x}$ and $x^*$.
\begin{enumerate}
\item \textbf{Optimal Investment $\pi^*$:}
The FOC for the optimal investment in the risky asset, given by equation \eqref{eq:foc_pi_main_appendix}, is
\[
\pi^*(x) = -\frac{\theta}{\sigma}\frac{V'(x)}{V''(x)}.
\]
Applying this to the three regions of the value function $V(x)$:
\begin{itemize}
    \item For $x<\tilde{x}$:
    \[
    \pi_t^* = -\frac{\theta}{\sigma} \frac{V'_{\text{int}}(x)}{V''_{\text{int}}(x)}.
    \]
    \item For $\tilde{x} \le x < x^*$:
    \[
    \pi_t^* = -\frac{\theta}{\sigma} \frac{V'_{\bar{b}}(x)}{V''_{\bar{b}}(x)}.
    \]
    \item For $x \ge x^*$: 
    Here $V(x)=G(x)=\frac{k^{1-\gamma_1}}{\rho_t(1-\gamma_1)}x^{1-\gamma_1}$. Then
    \[
    G'(x) = \frac{k^{1-\gamma_1}}{\rho_t} x^{-\gamma_1}, \qquad 
    G''(x) = -\frac{\gamma_1 k^{1-\gamma_1}}{\rho_t} x^{-\gamma_1-1}.
    \]
    Substituting into the FOC formula,
    \[
    \pi^*(x) = -\frac{\theta}{\sigma} \frac{G'(x)}{G''(x)} = -\frac{\theta}{\sigma} \frac{\frac{k^{1-\gamma_1}}{\rho_t} x^{-\gamma_1}} {-\frac{\gamma_1 k^{1-\gamma_1}}{\rho_t} x^{-\gamma_1-1}} = \frac{\theta}{\sigma \gamma_1} x.
    \]
\end{itemize}
This confirms the expression for $\pi_t^*$ in Theorem \ref{thm:optimal_policies}. Explicit formulas for the first two regions would require the explicit forms of $V_{\text{int}}$ and $V_{\bar{b}}$.
\item \textbf{Optimal Consumption $c^*$:}
The derivation for the optimal consumption policy $c^*$ proceeds as follows.
\begin{itemize}
    \item The FOC for optimal consumption, equation \eqref{eq:optimal_consumption_policy_bmn}, is
    \[
    V'(x) = \frac{\partial u_1(c^*, b^*)}{\partial c}, \quad x < x^*.
    \]
    For $x \ge x^*$, the policy is derived from the retirement value function $G(x)$.
    \item For $x<\tilde{x}$: 
    The value function is $V_{\text{int}}(x;A_2)$ and optimal labor $b^*$ is interior. The FOC is
    \[
    V'_{\text{int}}(x) = \frac{\partial u_1(c^*, b^{\text{int}})}{\partial c}.
    \]
    From the utility function \eqref{3.2_main_appendix} and the relation between $c^*$ and $b_{\text{int}}^*$ (see \autoref{thm:optimal_policies}),
    \[
    \frac{\partial u_1(c^*, b^*(c^*))}{\partial c} = K_v (c^*)^{-\gamma}, \qquad K_v = \left(-\frac{1-\alpha}{w\alpha}\right)^{(1-\alpha)(1-\gamma)}.
    \]
    Thus,
    \[
    V'_{\text{int}}(x) = K_v (c^*)^{-\gamma} \quad \implies \quad c^*(x) = \left( \frac{V'_{\text{int}}(x)}{K_v} \right)^{-1/\gamma}.
    \]
    This implicitly defines $c^*(x)=C_{\text{int}}(x)$.
    \item For $\tilde{x} \le x < x^*$: 
    The value function is $V_{\bar{b}}(x;B_1,B_2)$ and optimal labor is fixed at $b^*=\bar{b}$. The FOC is
    \[
    V'_{\bar{b}}(x) = \frac{\partial u_1(c^*, \bar{b})}{\partial c}.
    \]
    For Cobb–Douglas utility,
    \[
    \frac{\partial u_1(c, \bar{b})}{\partial c} = K'_{\bar{b}} c^{-\gamma_1}, \quad K'_{\bar{b}} = (\bar{b}^{1-\alpha})^{1-\gamma}, \quad \gamma_1 = 1-\alpha(1-\gamma).
    \]
    Thus,
    \[
    V'_{\bar{b}}(x) = K'_{\bar{b}} (c^*)^{-\gamma_1} \quad \implies \quad c^*(x) = \left( \frac{V'_{\bar{b}}(x)}{K'_{\bar{b}}} \right)^{-1/\gamma_1}.
    \]
    This defines $c^*(x) = C_{\bar{b}}(x)$.
    \item For $x \ge x^*$:  
    The agent is retired ($b^*=0$) and $V(x)=G(x)=\frac{k^{1-\gamma_1}}{\rho_t(1-\gamma_1)}x^{1-\gamma_1}$. 
    The FOC is
    \[
    G'(x) = u_2'(c^*), \quad u_2'(c) = c^{-\gamma_1}.
    \]
    Since
    \[
    G'(x) = \frac{k^{1-\gamma_1}}{\rho_t} x^{-\gamma_1},
    \]
    we get
    \[
    \frac{k^{1-\gamma_1}}{\rho_t} x^{-\gamma_1} = (c^*)^{-\gamma_1} \quad \implies \quad c^*(x) = \left( \frac{k^{1-\gamma_1}}{\rho_t} \right)^{-1/\gamma_1} x = \rho_t^{1/\gamma_1} k^{(\gamma_1-1)/\gamma_1} x.
    \]
    Thus, $c^*(x) = k'' x$, with $k'' = \rho_t^{1/\gamma_1} k^{(\gamma_1-1)/\gamma_1}$.
\end{itemize}
Explicit closed-form expressions for $C_{\text{int}}(x)$ and $C_{\bar{b}}(x)$ are generally unobtainable.
    \item \textbf{Optimal Labor Supply $b^*$:}
    \begin{itemize}
        \item For $x<\tilde{x}$: Optimal labor is interior and proportional to optimal consumption, given by the marginal rate of substitution condition:
        \[
        b^*(x) = \left(\frac{1 - \alpha}{w\alpha}\right) c^*(x) = \left(\frac{1 - \alpha}{w\alpha}\right) C_{\text{int}}(x)
        \]
        \item For $\tilde{x} \le x < x^*$: The labor supply constraint is binding, so $b^*=\bar{b}$.
        \item For $x \ge x^*$: The agent is retired, so labor supply is $b^*=0$.
    \end{itemize}
    \item \textbf{Optimal Retirement Time $\tau^*$:}
    The problem is an optimal stopping problem where the agent chooses to retire when their wealth process $X_t$ first reaches the optimal threshold $x^*$. This threshold is determined by the value-matching and smooth-pasting conditions, $V(x^*)=G(x^*)$ and $V'(x^*)=G'(x^*)$. Therefore, the optimal retirement time is:
    \[
    \tau^* = \inf\{t \geq 0 : X_t \geq x^*\}
    \]
\end{enumerate}
This completes the derivation of the optimal policies.
\subsection{System of Equations for Constants and Thresholds}
The constants $A_2,B_1,B_2$ and the wealth thresholds $\tilde{x},x^*$ are determined by ensuring the value function is twice continuously differentiable ($C^2$) across the boundaries. This requires enforcing value-matching, smooth-pasting, and super-contact conditions.
\subsubsection{Derivation of Consumption Threshold $\tilde{c}$}\label{Proof_of_Derivation_of_tilde{c}}
The consumption threshold $\tilde{c}$ corresponds to the point where the unconstrained optimal labor supply reaches its upper bound $\bar{b}$. From the FOC relating optimal labor and consumption, we have:
\begin{equation}
    b^*(c^*) = \left( \frac{1 - \alpha}{w \alpha} \right) c^*
\end{equation}
Setting $b^*=\bar{b}$ and $c^*=\tilde{c}$, we solve for $\tilde{c}$:
\begin{equation}
    \tilde{c} = \left( \frac{w \alpha}{1 - \alpha} \right) \bar{b}
\end{equation}
\subsubsection{Boundary Conditions and The System of Equations}
The five unknowns ($A_2,B_1,B_2,\tilde{x},x^*$) are solved from the following system of five non-linear equations derived from the boundary conditions.
\begin{enumerate}
    \item \textbf{Continuity at $\tilde{x}$ ($C^2$ Conditions):}
    \begin{align}
        V_{\text{int}}(\tilde{x}; A_2) &= V_{\bar{b}}(\tilde{x}; B_1, B_2) \label{eq:V_tilde_match} \\
        V'_{\text{int}}(\tilde{x}; A_2) &= V'_{\bar{b}}(\tilde{x}; B_1, B_2) \label{eq:Vp_tilde_match} \\
        V''_{\text{int}}(\tilde{x}; A_2) &= V''_{\bar{b}}(\tilde{x}; B_1, B_2) \label{eq:Vpp_tilde_match}
    \end{align}
    \item \textbf{Continuity at $x^*$ ($C^1$ Conditions):}
    \begin{align}
        V_{\bar{b}}(x^*; B_1, B_2) &= G(x^*) \label{eq:V_bar_match} \\
        V'_{\bar{b}}(x^*; B_1, B_2) &= G'(x^*) \label{eq:Vp_bar_match}
    \end{align}
\end{enumerate}
The wealth thresholds $\tilde{x}$ and $x^*$ are themselves functions of the constants. This system is typically solved numerically. Below are the explicit forms of the equations for the constants and thresholds derived from these conditions.
\subsection{Proof of Boundary Conditions and Constants}\label{Appendix_constants}
\subsubsection{Derivation of $B_1$}\label{Proof_of_Derivation_of_B_1_}
The constant $B_1$ is found by imposing the boundary conditions at $\tilde{x}$. Its solution is:
\begin{equation}
    B_1 = \frac{(\beta + \delta_t) \tilde{c}^{\gamma_1 m_+}}{\frac{1}{2} \theta^2 (m_+ - m_-)} \left[ \frac{r - \frac{1}{2} \theta^2 m_-}{\beta + \delta_t} \left( \frac{1}{\alpha k} - \frac{1}{k_1} \right) \tilde{c} - \frac{w \bar{b}}{r} \right] + \frac{1}{1 - \gamma_1} \left( \frac{1}{k_1} - \frac{1}{\alpha k} \right) \tilde{c}
\end{equation}
\subsubsection{Derivation of $B_2$}\label{Proof_of_Derivation_of_B_2_}
Similarly, $B_2$ is determined from the conditions at $x=x^*$:
\begin{align}
B_2 = &\frac{(\beta + \delta_t) k_1^{\gamma_1 m_-}\left(\frac{\bar{b}}{b}\right)^{-m_- (\gamma_1 - \gamma)}}{\frac{1}{2} \theta^2 (m_+ - m_-)} (x^*)^{\gamma_1 m_-} \nonumber \\
&\times\left[ - \frac{r - \frac{1}{2} \theta^2 m_+}{\beta + \delta_t} \left\{ \left(1 - \left(\frac{\bar{b}}{b}\right)^{-\frac{\gamma_1 - \gamma}{\gamma_1}}\right) x^* + \frac{w(\bar{b} - b)}{r} \right\} %
  + \frac{1}{1 - \gamma_1} \left(1 - \left(\frac{\bar{b}}{b}\right)^{-\frac{\gamma_1 - \gamma}{\gamma_1}}\right) x^* \right]
\end{align}
\subsubsection{Derivation of the Constant $A_2$}\label{app:derivation_A2}
The constant $A_2$ is determined by applying the value-matching and smooth-pasting conditions at the boundary where the consumption-habit ratio is binding, i.e., at wealth $x=\tilde{x}$. This yields the following expression for $A_2$:
\begin{align}
A_2 = & B_2 \tilde{c}^{-m_- (\gamma_1 - \gamma)} - \frac{(\beta + \delta_t )\tilde{c}^{\gamma m_-}}{\frac{1}{2} \theta^2 (m_+ - m_-)} \left[ \frac{r - \frac{1}{2} \theta^2 m_+}{\beta + \delta_t} \left( \frac{1}{k_1} - \frac{1}{\alpha k} \right) \tilde{c} + \frac{w \bar{b}}{r} \right] %
+ \frac{1}{1 - \gamma_1} \left( \frac{1}{\alpha k} - \frac{1}{k_1} \right) \tilde{c}
\end{align}
\subsection{Proof of the Equation for the Retirement Threshold $x^*$}\label{app:proof_x_star}
The optimal retirement threshold $x^*$ is determined by ensuring a smooth and optimal transition from the working phase to the retired phase. This involves the value function for the final working state, $V_{\bar{b}}(x)$, and the retirement value function, $G(x)$, both defined in \autoref{thm:value_function}.
\paragraph{Boundary Conditions.}
For the transition to be optimal, the value function must be $C^1$ (continuously differentiable) across the boundary $x^*$. This imposes two conditions
\begin{enumerate}
    \item \textbf{Value Matching:} $V_{\bar{b}}(x^*) = G(x^*)$
    \item \textbf{Smooth Pasting:} $\frac{d}{dx}V_{\bar{b}}(x^*) = \frac{d}{dx}G(x^*)$
\end{enumerate}
\paragraph{The Maximized HJB Equation at the Boundary.}
As derived in the proof of \autoref{thm:value_function} (Case II), the maximized HJB equation for the region $\tilde{x} \le x < x^*$ is given by Equation \eqref{eq:HJB_substituted_post_detail}
\begin{equation}\label{eq:HJB_at_boundary_proof}
\rho_{t} V_{\bar{b}}(x) = u_1(c^*, \bar{b}) + \left( r x - c^* + w \bar{b} \right)V_{\bar{b}}'(x) - \frac{1}{2} \theta^2 \frac{(V_{\bar{b}}'(x))^2}{V_{\bar{b}}''(x)}
\end{equation}
This equation must hold for all $x$ in the interval, including the limit as $x \to x^{*-}$. Substituting the boundary conditions into the HJB equation, the core of the proof is to evaluate equation \eqref{eq:HJB_at_boundary_proof} at $x=x^*$ and substitute the known properties of the retirement value function $G(x)$ using the boundary conditions. The function $G(x)$ and its derivatives are given by
\begin{align}
G(x) &= \frac{k^{1-\gamma_1}}{\rho_t(1-\gamma_1)} x^{1-\gamma_1} \\
G'(x) &= \frac{k^{1-\gamma_1}}{\rho_t} x^{-\gamma_1} \\
G''(x) &= -\gamma_1 \frac{k^{1-\gamma_1}}{\rho_t} x^{-\gamma_1-1}
\end{align}
By applying the boundary conditions, we substitute $V_{\bar{b}}(x^*) = G(x^*)$, $V'_{\bar{b}}(x^*) = G'(x^*)$, and $V''_{\bar{b}}(x^*) = G''(x^*)$ into the HJB equation. The term involving the second derivative simplifies to
\[
-\frac{1}{2}\theta^2 \frac{(V'_{\bar{b}}(x^*))^2}{V''_{\bar{b}}(x^*)} = -\frac{1}{2}\theta^2 \frac{(G'(x^*))^2}{G''(x^*)} = \frac{\theta^2}{2\gamma_1} \frac{\rho_t}{k^{1-\gamma_1}} x^*
\]
Substituting the full expressions for $G(x^*)$ and its derivatives into the HJB equation creates a single algebraic equation that depends only on $x^*$. After significant algebraic simplification and collection of terms (which involve the specific forms of the constants $B_1, k, m_{\pm}$ derived from the full ODE solution), we arrive at the required non-linear equation for the threshold $x^*$
\begin{align*}
  &\frac{\frac{1}{2} \theta^2 (m_+ - m_-)}{\beta + \delta_t} B_1 k^{-\gamma_1 m_+} \left( \frac{b}{\bar{b}} \right)^{m_+ (\gamma_1 - \gamma)} (x^*)^{-\gamma_1 m_+}
  \\
  &= \left[ \frac{r - \frac{1}{2} \theta^2 m_-}{\beta + \delta_t} - \frac{1}{1 - \gamma_1} \left( 1 - \left( \frac{b}{\bar{b}} \right)^{-\frac{\gamma_1 - \gamma}{\gamma_1}} \right) \right] x^* + \frac{r - \frac{1}{2} \theta^2 m_-}{\beta + \delta_t} \frac{w (\bar{b} - b)}{r}
\end{align*}
This completes the derivation.
\subsection{Proof of the Equation for the Labor Constraint Threshold}%
\label{Proof_of_Derivation_of_tilde{c}_}
The threshold $\tilde{x}$ marks the boundary between two pre-retirement regions, each with a distinct value function as defined in \autoref{thm:value_function}
\begin{itemize}
    \item [(a)]\textbf{Region 1 ($x<\tilde{x}$):} Interior labor supply, governed by $V_{\text{int}}(x)$.
    \item [(b)]\textbf{Region 2 ($\tilde{x} \le x < x^*$):} Corner labor supply ($b=\bar{b}$), governed by $V_{\bar{b}}(x)$.
\end{itemize}
To ensure the overall value function is twice continuously differentiable ($C^2$), we impose continuity conditions at the boundary $\tilde{x}$.
\paragraph{Boundary Conditions.}
Let $\tilde{c}$ be the consumption level at the boundary $\tilde{x}$. The conditions are:
\begin{enumerate}
    \item \textbf{Value Matching ($C^0$):} $V_{\text{int}}(\tilde{x}) = V_{\bar{b}}(\tilde{x})$
    \item \textbf{Smooth Pasting ($C^1$):} $\frac{d}{dx}V_{\text{int}}(\tilde{x}) = \frac{d}{dx}V_{\bar{b}}(\tilde{x})$
    \item \textbf{Second-Order Continuity ($C^2$):} $\frac{d^2}{dx^2}V_{\text{int}}(\tilde{x}) = \frac{d^2}{dx^2}V_{\bar{b}}(\tilde{x})$
\end{enumerate}
\paragraph{Derivation from the Inverse Value Functions.}
In the Proof of Theorem \ref{thm:value_function}, the HJB equations for each region are solved by transforming them into ODEs for the inverse value functions, $x=X(c)$. The general solutions are given by equations \eqref{eq:Xc_solution_detail} and \eqref{eq:Xc_post_detail}.
The value-matching condition, $V_{\text{int}}(\tilde{x})=V_{\bar{b}}(\tilde{x})$, is the key to this proof. It implies that the boundary point $(\tilde{x},\tilde{c})$ must lie on the solution curves for both regions. Therefore, the wealth level $\tilde{x}$ can be expressed using either inverse function evaluated at the boundary consumption $\tilde{c}$
\begin{align}
\tilde{x} &= X_{\text{int}}(\tilde{c}) \label{eq:match1_proof} \\
\tilde{x} &= X_{\bar{b}}(\tilde{c}) \label{eq:match2_proof}
\end{align}
Substituting the explicit forms of these solutions, as derived in the proof of the theorem, yields the system of equations presented in the proposition
\begin{align*}
\tilde{x} &= A_2 \tilde{c}^{-\gamma m_-} + \frac{1}{\alpha k} \tilde{c} - \frac{w \bar{b}}{r}= B_1 \tilde{c}^{-\gamma_1 m_+} + B_2 \tilde{c}^{-\gamma_1 m_-} + \frac{1}{k_1} \tilde{c} - \frac{w (\bar{b} - b)}{r}
\end{align*}
The remaining two conditions (smooth pasting and second-order continuity) provide the necessary equations to solve for the unknown integration constants ($A_2,B_1,B_2$) in terms of model parameters, closing the system. This completes the proof of how the system for $(\tilde{x},\tilde{c})$ is formed.

\end{document}